\documentclass[a4paper,USenglish,cleveref,numberwithinsect]{lipics-v2019}

    \usepackage{enumitem}
    \usepackage{microtype}

    \title{Preprocessing to Reduce the Search Space: Antler Structures for Feedback Vertex Set}

    \titlerunning{Preprocessing to Reduce the Search Space: Antler Structures for FVS}

    \author{Huib Donkers}{Eindhoven University of Technology, The Netherlands}{h.t.donkers@tue.nl}{0000-0002-2767-8140}{}

    \author{Bart M.\,P. Jansen}{Eindhoven University of Technology, The Netherlands}{b.m.p.jansen@tue.nl}{0000-0001-8204-1268}{This project has received funding from the European Research Council (ERC) under the European Union's Horizon 2020 research and innovation programme (grant agreement No 803421, ReduceSearch). \\ \includegraphics[height=2cm]{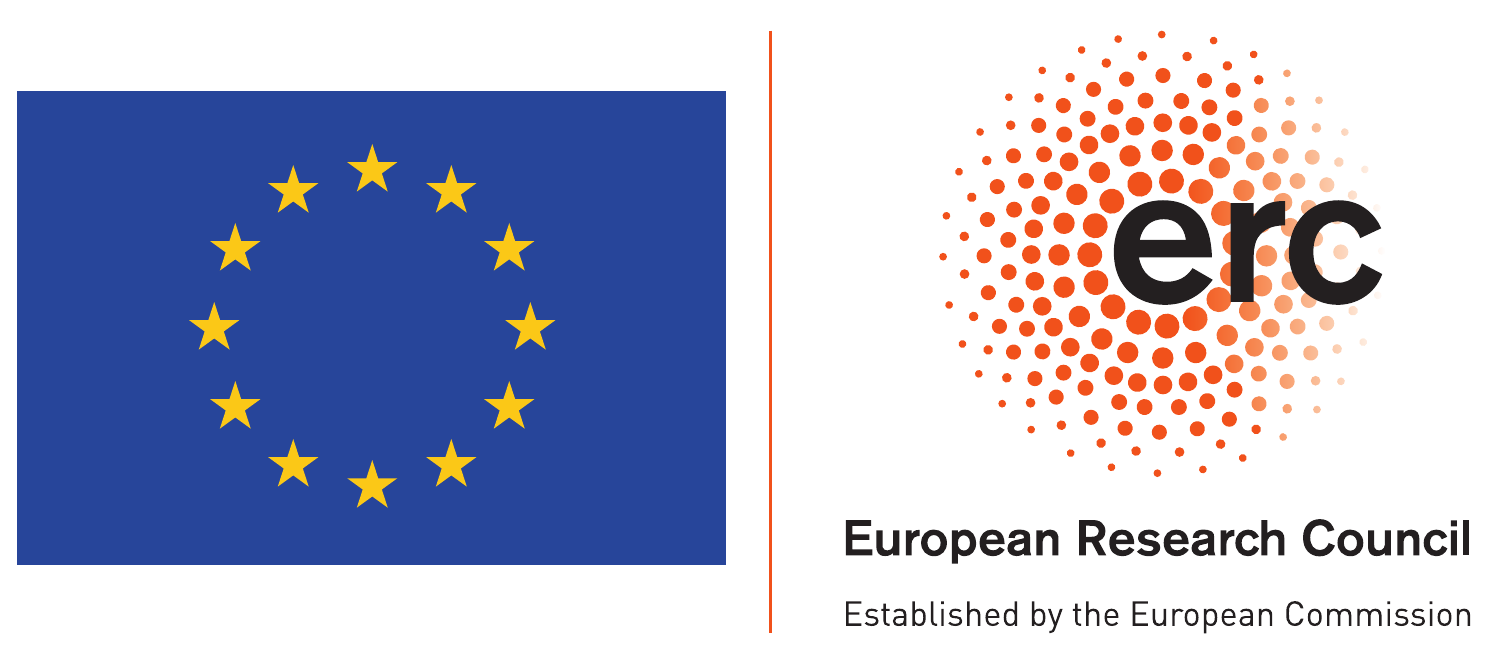}}

    \authorrunning{H.\,T. Donkers and B.\,M.\,P. Jansen}

    \Copyright{Huib Donkers and Bart M.\,P. Jansen}

		\ccsdesc[100]{Theory of computation~Graph algorithms analysis}
		\ccsdesc[100]{Theory of computation~Parameterized complexity and exact algorithms}

    \keywords{kernelization, preprocessing, feedback vertex set, graph decomposition}

    \category{}

    \relatedversion{An extended abstract of this work appeared in the Proceedings of the 47th Workshop on Graph-Theoretical Concepts in Computer Science, WG 2021}

    \supplement{}

    \funding{}

    \acknowledgements{In memory of Rolf Niedermeier and his love for the art of problem parameterization.}

    \EventEditors{John Q. Open and Joan R. Access}
    \EventNoEds{2}
    \EventLongTitle{42nd Conference on Very Important Topics (CVIT 2016)}
    \EventShortTitle{CVIT 2016}
    \EventAcronym{CVIT}
    \EventYear{2016}
    \EventDate{December 24--27, 2016}
    \EventLocation{Little Whinging, United Kingdom}
    \EventLogo{}
    \SeriesVolume{42}
    \ArticleNo{23}
    \nolinenumbers 
    \hideLIPIcs  

    \theoremstyle{plain}
    
    \newtheorem{observation}[theorem]{Observation}
    
    \newtheorem{reduction}{Reduction Rule}
    
    \crefname{operation}{Operation}{Operations}
    \crefname{condition}{Condition}{Conditions}
    \crefname{property}{Property}{Properties}
    \crefname{reduction}{Reduction Rule}{Reduction Rules}
    \crefname{claim}{Claim}{Claims}

\crefalias{step}{enumi}
\crefname{step}{Step}{Steps}

\makeatletter
\let\oldclaimproof\claimproof
\def\claimproof{\oldclaimproof\renewcommand\qedsymbol{\textcolor{lipicsGray}{\ensuremath{\vartriangleleft}}}}
\makeatother

\usepackage{algorithm2e}

\newcommand{\inv}[1]{\ensuremath{#1^{-1}}}
\DeclareMathOperator{\degree}{deg} 
\DeclareMathOperator{\vc}{\textsc{vc}}
\DeclareMathOperator{\fvs}{\textsc{fvs}} 
\DeclareMathOperator{\fvsLP}{\mathsf{fvs_{LP}}}

\newcommand{\cF}{\ensuremath{\mathsf{\dot{F}}}\xspace}
\newcommand{\cC}{\ensuremath{\mathsf{\dot{C}}}\xspace}
\newcommand{\cR}{\ensuremath{\mathsf{\dot{R}}}\xspace}

\newcommand{\Oh}{\ensuremath{\mathcal{O}}\xspace}

\newcommand{\scOCT}{\textsc{Odd Cycle Transversal}\xspace}
\newcommand{\scDFVS}{\textsc{Directed Feedback Vertex Set}\xspace}
\newcommand{\cnfsat}{\textsc{CNF-SAT}\xspace}
\newcommand{\scVC}{\textsc{Vertex Cover}\xspace}
\newcommand{\scFVS}{\textsc{Feedback Vertex Set}\xspace}

\newcommand{\crown}{\ensuremath{\mathsf{crown}}\xspace}
\newcommand{\head}{\ensuremath{\mathsf{head}}\xspace}

\newcommand{\antler}{\ensuremath{\mathsf{antler}}\xspace}

\newcommand{\setN}{\ensuremath{\mathbb{N}}\xspace}

\newcommand{\PNEQNP}{\ensuremath{\mathsf{P}}~$\neq$~\ensuremath{\mathsf{NP}}\xspace}
\newcommand{\NP}{\ensuremath{\mathsf{NP}}\xspace}
\newcommand{\Wone}{\ensuremath{\mathsf{W[1]}}}
\newcommand{\NPhard}{\NP-hard\xspace}
\newcommand{\Wonehard}{\Wone-hard\xspace}
\newcommand{\NPcomplete}{\NP-complete\xspace}

\newcommand{\scBWoneAntlerDetection}{\textsc{Bounded-Width 1-Antler Detection}\xspace}
\newcommand{\scMultColoredClique}{\textsc{Multicolored Clique}\xspace}

 \newcommand{\defparproblem}[4]{
 \vspace{1mm}
 \noindent\fbox{
 \begin{minipage}{0.96\textwidth}
 \textsc{#1} \\
 {\bf{Input:}} #2 \\
 {\bf{Parameter:}} #3 \\
 {\bf{Question:}} #4
 \end{minipage}
 }
 \vspace{1mm}
 }

\begin{document}

\maketitle

\begin{abstract}
The goal of this paper is to open up a new research direction aimed at understanding the power of preprocessing in speeding up algorithms that solve NP-hard problems exactly. We explore this direction for the classic \textsc{Feedback Vertex Set} problem on undirected graphs, leading to a new type of graph structure called \emph{antler decomposition}, which identifies vertices that belong to an optimal solution. It is an analogue of the celebrated \emph{crown decomposition} which has been used for \textsc{Vertex Cover}. We develop the graph structure theory around such decompositions and develop fixed-parameter tractable algorithms to find them, parameterized by the number of vertices for which they witness presence in an optimal solution. This reduces the search space of fixed-parameter tractable algorithms parameterized by the solution size that solve \textsc{Feedback Vertex Set}.
\end{abstract}

\clearpage

\section{Introduction} \label{antler:sec:intro}
The goal of this paper is to open up a new research direction aimed at understanding the power of preprocessing in speeding up algorithms that solve NP-hard problems exactly~\cite{Fellows06,GuoN07a}. In a nutshell, this new direction can be summarized as: how can an algorithm identify part of an optimal solution in an efficient preprocessing phase? We explore this direction for the classic~\cite{Karp72} \textsc{Feedback Vertex Set} problem on undirected graphs, leading to a new graph structure called \emph{antler} which reveals vertices that belong to an optimal feedback vertex set.

We start by motivating the need for a new direction in the theoretical analysis of preprocessing. The use of preprocessing, often via the repeated application of reduction rules, has long been known~\cite{AchterbergBGRW20,Achterberg2013,Quine52} to speed up the solution of algorithmic tasks in practice. The introduction of the framework of parameterized complexity~\cite{DowneyF99} in the 1990s made it possible to also analyze the power of preprocessing \emph{theoretically}, through the notion of kernelization. It applies to \emph{parameterized decision problems}~$\Pi \subseteq \Sigma^* \times \mathbb{N}$, in which every instance~$x \in \Sigma^*$ has an associated integer parameter~$k$ which captures one dimension of its complexity. For \textsc{Feedback Vertex Set}, typical choices for the parameter include the size of the desired solution or structural measures of the complexity of the input graph. A kernelization for a parameterized problem~$\Pi$ is then a polynomial-time algorithm that reduces any instance with parameter value~$k$ to an equivalent instance, of the same problem, whose total size is bounded by~$f(k)$ for some computable function~$f$ of the parameter alone. The function~$f$ is the \emph{size} of the kernelization.

A substantial theoretical framework has been built around the definition of kernelization~\cite{CyganFKLMPPS15,DowneyF13,FlumG06,FominLSZ19,GuoN07a}. It includes deep techniques for obtaining kernelization algorithms~\cite{BodlaenderFLPST16,FominLMS12,KratschW20,PilipczukPSL18}, as well as tools for ruling out the existence of small kernelizations~\cite{BodlaenderJK14,DellM14,Drucker15,FortnowS11,HermelinKSWW15} under complexity-theoretic hypotheses. This body of work gives a good theoretical understanding of polynomial-time data compression for NP-hard problems. 

However, we argue that these results on kernelization \emph{do not} explain the often exponential speed-ups (e.g.~\cite{AchterbergBGRW20},~\cite[Table 6]{AkibaI16}) caused by applying effective preprocessing steps to non-trivial algorithms. Why not? A kernelization algorithm guarantees that the input \emph{size} is reduced to a function of the parameter~$k$; but the running time of modern parameterized algorithms for NP-hard problems is not exponential in the total input size. Instead, fixed-parameter tractable (FPT) algorithms have a running time that scales polynomially with the input size, and which only depends exponentially on a problem parameter such as the solution size or treewidth. Hence an exponential speed-up of such algorithms cannot be explained by merely a decrease in input size, but only by a decrease in the \emph{parameter}! 

We therefore propose the following novel research direction: to investigate how preprocessing algorithms can decrease the parameter value (and hence search space) of FPT algorithms, in a theoretically sound way. It is nontrivial to phrase meaningful formal questions in this direction. To illustrate this difficulty, note that strengthening the definition of kernelization to ``a preprocessing algorithm that is guaranteed to always output an equivalent instance of the same problem with a strictly smaller parameter'' is useless. Under minor technical assumptions, such an algorithm would allow the problem to be solved in polynomial time by repeatedly reducing the parameter, and solving the problem using an FPT or XP algorithm once the parameter value becomes constant. Hence NP-hard problems do not admit such parameter-decreasing algorithms. To formalize a meaningful line of inquiry, we take our inspiration from the \textsc{Vertex Cover} problem, the fruit fly of parameterized algorithms.

A rich body of theoretical and applied algorithmic research has been devoted to the exact solution of the \textsc{Vertex Cover} problem~\cite{AkibaI16,DzulfikarFH19,HespeLSS19,Hespe0S19}. A standard 2-way branching algorithm can test whether a graph~$G$ has a vertex cover of size~$k$ in time~$\Oh(2^k (n+m))$, which can be improved by more sophisticated techniques~\cite{ChenKX10}. The running time of the algorithm scales linearly with the input size, and exponentially with the size~$k$ of the desired solution. This running time suggests that to speed up the algorithm by a factor~$1000$, one either has to decrease the input size by a factor~$1000$, or decrease~$k$ by~$\log_2 (1000) \approx 10$. 

It turns out that state-of-the-art preprocessing strategies for \textsc{Vertex Cover} indeed often \emph{succeed} in decreasing the size of the solution that the follow-up algorithm has to find, by means of crown-reduction~\cite{Abu-KhzamFLS07,ChorFJ04,Fellows03}, the intimately related Nemhauser-Trotter reduction based on the linear-programming relaxation~\cite{NemhauserT75}, and additional ad-hoc rules such as \emph{funnel}, \emph{unconfined}, and~\emph{twin}~\cite[\S 8.3.2]{AkibaI16}. Recall that a vertex cover in a graph~$G$ is a set~$S \subseteq V(G)$ such that each edge has at least one endpoint in~$S$. Observe that if~$H \subseteq V(G)$ is a set of vertices with the property that there exists a minimum vertex cover of~$G$ containing all of~$H$, then~$G$ has a vertex cover of size~$k$ if and only if~$G - H$ has a vertex cover of size~$k - |H|$. Therefore, if a preprocessing algorithm can identify a set of vertices~$H$ which are guaranteed to belong to an optimal solution, then it can effectively reduce the parameter of the problem by restricting to a search for a solution of size~$k - |H|$ in~$G-H$. 

A \emph{crown decomposition} (cf.~\cite{Abu-KhzamCFLSS04,ChorFJ04,Fellows03},~\cite[\S 2.3]{CyganFKLMPPS15},~\cite[\S 4]{FominLSZ19}) of a graph~$G$ serves exactly this purpose. It consists of two disjoint vertex sets~$(\head,\crown)$, such that~$\crown$ is a non-empty independent set whose neighborhood is contained in~$\head$, and such that the graph~$G[\head \cup \crown]$ has a matching~$M$ of size~$|\head|$. As~$\crown$ is an independent set, the matching~$M$ assigns to each vertex of~$\head$ a private neighbor in~$\crown$. It certifies that any vertex cover in~$G$ contains at least~$|\head|$ vertices from~$\head \cup \crown$, and as~$\crown$ is an independent set with~$N_G(\crown) \subseteq \head$, a simple exchange argument shows there is indeed an optimal vertex cover in~$G$ containing all of~$\head$ and none of~$\crown$. Since there is a polynomial-time algorithm to find a crown decomposition if one exists~\cite[Thm.~11--12]{Abu-KhzamFLS07}, this yields the following preprocessing guarantee for \textsc{Vertex Cover}: if the input instance~$(G,k)$ has a crown decomposition~$(\head,\crown)$, then a polynomial-time algorithm can reduce the problem to an equivalent one with parameter at most~$k - |\head|$, thereby giving a formal guarantee on reduction in the parameter based on the structure of the input.\footnote{The effect of the crown reduction rule can also be theoretically explained by the fact that interleaving basic 2-way branching with exhaustive crown reduction yields an algorithm whose running time is only exponential in the \emph{gap} between the size of a minimum vertex cover and the cost of an optimal solution to its linear-programming relaxation~\cite{LokshtanovNRRS14}. However, this type of result cannot be generalized to \textsc{Feedback Vertex Set} since it is already NP-complete to determine whether there is a feedback vertex set whose size matches the cost of the linear-programming relaxation (\cref{cor:lp-relax}).}

As the first step of our proposed research program into parameter reduction (and thereby, search space reduction) by a preprocessing phase, we present a graph decomposition for \textsc{Feedback Vertex Set} which can identify vertices~$S$ that belong to an optimal solution; and which therefore facilitate a reduction from finding a solution of size~$k$ in graph~$G$, to finding a solution of size~$k - |S|$ in~$G - S$. While there has been a significant amount of work on kernelization for \textsc{Feedback Vertex Set}~\cite{BodlaenderD10,BurrageEFLMR06,Iwata17,JansenRV14,Thomasse10}, the corresponding preprocessing algorithms do not succeed in finding vertices that belong to an optimal solution, other than those for which there is a self-loop or those which form the center a flower (consisting of~$k+1$ otherwise vertex-disjoint cycles~\cite{BodlaenderD10,BurrageEFLMR06,Thomasse10}, or a technical relaxation of this notion~\cite{Iwata17}). In particular, apart from the trivial self-loop rule, earlier preprocessing algorithms can only conclude a vertex~$v$ belongs to all optimal solutions (of a size~$k$ which must be given in advance) if they find a suitable packing of cycles witnessing that solutions without~$v$ must have size larger than~$k$. In contrast, our argumentation will be based on \emph{local} exchange arguments, which can be applied independently of the global solution size~$k$.

We therefore introduce a new graph decomposition for preprocessing \scFVS. To motivate it, we distill the essential features of a crown decomposition. Effectively, a crown decomposition of~$G$ certifies that~$G$ has a minimum vertex cover containing all of~$\head$, because (i)~any vertex cover has to pick at least~$|\head|$ vertices from~$\head \cup \crown$, as the matching~$M$ certifies that~$\vc(G[\head \cup \crown]) \geq |\head|$, while (ii)~any minimum vertex cover~$S'$ in~$G - (\head \cup \crown)$ yields a minimum vertex cover~$S' \cup \head$ in~$G$, since~$N_G(\crown) \subseteq \head$ and~$\crown$ is an independent set. To obtain similar guarantees for \scFVS, we need a decomposition to supply disjoint vertex sets~$(\head,\antler)$ such that:
\begin{enumerate}
 \item \label[condition]{condition:small-head} any minimum feedback vertex set in~$G$ contains at least~$|\head|$ vertices from~$\head \cup \antler$, and
 \item \label[condition]{condition:remaining-sol} any minimum feedback vertex set~$S'$ in~$G - (\head \cup \antler)$ yields a minimum feedback vertex set~$S' \cup \head$ in~$G$.
\end{enumerate}

To achieve \cref{condition:small-head}, it suffices for~$G[\head \cup \antler]$ to contain a set of~$|\head|$ vertex-disjoint cycles; to achieve \cref{condition:remaining-sol}, it suffices for~$G[\antler]$ to be acyclic, with each tree~$T$ of the forest~$G[\antler]$ connected to the remainder~$V(G) \setminus (\head \cup \antler)$ by at most one edge (implying that all cycles through~$\antler$ intersect~$\head$). We call such a decomposition a 1-antler. Here \emph{antler} refers to the shape of the forest~$G[\antler]$, which no longer consists of isolated spikes of a crown (see \cref{fig:CrownAntler}). The prefix~$1$ indicates it is the simplest type of antler; we present a  generalization later.
An antler is \emph{non-empty} if~$\head \cup \antler \neq \emptyset$, and the \emph{width} of the antler is defined to be~$|\head|$.

\begin{figure}
	\centering
	\includegraphics{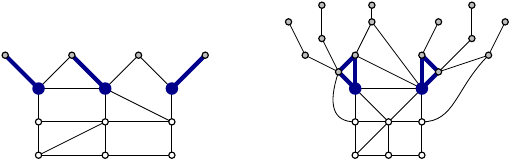}
	\caption{Graph structures showing there is an optimal solution containing all blue vertices and no gray vertices, certified by the blue subgraph. Left: Crown decomposition for \scVC. Right: Antler for \scFVS. For legibility, the number of edges in the drawing has been restricted. It therefore has vertices of degree at most~$2$, which makes the graph it reducible by standard reduction rules; but adding all possible edges between gray and blue vertices leads to a structure of minimum degree at least three which is still a 1-antler.}
	\label{fig:CrownAntler}
\end{figure}

Unfortunately, assuming \PNEQNP there is no \emph{polynomial-time} algorithm that always outputs a non-empty 1-antler if one exists. We prove this in \cref{sec:np-hard}. However, for the purpose of making a preprocessing algorithm that reduces the search space, we can allow FPT time in a parameter such as~$|\head|$ to find a decomposition. Each fixed choice of~$|\head|$ would then correspond to a reduction rule which identifies a small ($|\head|$-sized) part of an optimal feedback vertex set, for which there is a simple certificate for it being part of an optimal solution. Such a reduction rule can then be iterated in the preprocessing phase, thereby potentially decreasing the target solution size (and search space) by an arbitrarily large amount. Hence we consider the parameterized complexity of testing whether a graph admits a non-empty 1-antler with~$|\head| = k$, parameterized by~$k$. On the one hand, we show this problem to be W[1]-hard in \cref{sec:w1-hard}. This hardness turns out to be a technicality based on the forced bound on~$|\head|$, though. We provide the following FPT algorithm which yields a search-space reducing preprocessing step.

\begin{theorem} \label{thm:1:antler}
There is an algorithm that runs in~$2^{\Oh(k^5)} \cdot n^{\Oh(1)}$ time that, given a multigraph~$G$ on~$n$ vertices and integer~$k$, either correctly determines that~$G$ does not admit a non-empty 1-antler of width~$k$, or outputs a set~$S$ of at least~$k$ vertices such that there exists an optimal feedback vertex set in~$G$ containing all vertices of~$S$.
\end{theorem}

Hence if the input graph admits a non-empty 1-antler of width~$k$, the algorithm is guaranteed to find at least~$k$ vertices that belong to an optimal feedback vertex set, thereby reducing the search space.

Based on this positive result, we go further and generalize our approach beyond 1-antlers. For a 1-antler~$(\head, \antler)$ in~$G$, the set of~$|\head|$ vertex-disjoint cycles in~$G[\head \cup \antler]$ forms a very simple certificate that any feedback vertex set of~$G$ contains at least~$|\head|$ vertices from~$\head \cup \antler$. We can generalize our approach to identify part of an optimal solution, by allowing more complex certificates of optimality. The following interpretation of a 1-antler is the basis of the generalization: for a 1-antler~$(\head, \antler)$ in~$G$, there is a subgraph~$G'$ of~$G[\head \cup \antler]$ (formed by the~$|\head|$ vertex-disjoint cycles) such that~$V(G') \supseteq \head$ and~$\head$ is an optimal feedback vertex set of~$G'$; and furthermore this subgraph~$G'$ is simple because all its connected components, being cycles, have a feedback vertex set of size~$1$. For an arbitrary integer~$z$, we therefore define a $z$-antler in a multigraph~$G$ as a pair of disjoint vertex sets~$(\head, \antler)$ such that:
\begin{enumerate}
 \item the graph~$G[\head \cup \antler]$ has a subgraph~$G'$ for which~$\head$ is an optimal feedback vertex set and with each component of~$G'$ having a feedback vertex set of size at most~$z$ (this certifies that any minimum feedback vertex set in~$G$ contains at least~$|\head|$ vertices from~$\head \cup \antler$); and
 \item the graph~$G[\antler]$ is acyclic, with each tree~$T$ of the forest~$G[\antler]$ connected to the remainder~$V(G) \setminus (\head \cup \antler)$ by at most one edge. This second condition is not changed compared to a 1-antler.
\end{enumerate}
\begin{figure}[tb]
 \centering
 \includegraphics{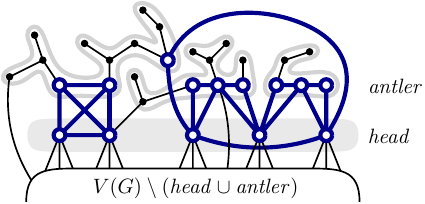}
 \caption{An example of a $3$-antler. The subgraph~$G'$, marked in blue, has a feedback vertex number of~$5$ showing that any feedback vertex set of~$G$ contains at least~$|\head| = 5$ vertices from~$\head \cup \antler$. Each connected component of~$G'$ has a feedback vertex set of size at most~$3$. The subgraph~$G[\antler]$ is acyclic and each of its connected components has at most one edge to~$V(G) \setminus (\head \cup \antler)$.}
 \label{fig:antler:3-antler}
\end{figure}
See \cref{fig:antler:3-antler} for an example of a $3$-antler of width~$5$. Our main result is the following.

\begin{theorem} \label{thm:z:antler}
There is an algorithm that runs in~$2^{\Oh(k^5 z^2)} \cdot n^{\Oh(z)}$ time that, given an undirected multigraph~$G$ on~$n$ vertices and integers~$k \geq z \geq 0$, either correctly determines that~$G$ does not admit a non-empty $z$-antler of width~$k$, or outputs a set~$S$ of at least~$k$ vertices such that there exists an optimal feedback vertex set in~$G$ containing all vertices of~$S$.
\end{theorem}

In fact, we prove a slightly stronger statement. If a graph~$G$ can be reduced to a graph~$G'$ by iteratively removing~$z$-antlers, each of width at most~$k$, and the sum of the widths of this sequence of antlers is~$t$, then we can find in time~$f(k,z) \cdot n^{\Oh(z)}$ a subset of at least~$t$ vertices of~$G$ that belong to an optimal feedback vertex set. This implies that if a complete solution to \textsc{Feedback Vertex Set} can be assembled by iteratively combining $\Oh(1)$-antlers of width at most~$k$, then the entire solution can be found in time~$f'(k) \cdot n^{\Oh(1)}$. Hence our work uncovers a new parameterization in terms of the complexity of the solution structure, rather than its size, in which \textsc{Feedback Vertex Set} is fixed-parameter tractable.

Our algorithmic results are based on a combination of graph reduction and color coding~\cite{AlonYZ95} (more precisely, its derandomization via the notion of universal sets). We use reduction steps inspired by the kernelization algorithms~\cite{BodlaenderD10,Thomasse10} for \textsc{Feedback Vertex Set} to bound the size of~$\antler$ in the size of~$\head$, by analyzing an intermediate structure called \emph{feedback vertex cut}. After such reduction steps, the size of the entire structure we are trying to find can be bounded in terms of the parameter~$k$. We then use color coding~\cite{AlonYZ95} to identify antler structures. A significant amount of effort goes into proving that the reduction steps preserve antler structures and the optimal solution size.

The remainder of the paper is organized as follows. After presenting preliminaries on graphs and sets in \cref{antler:sec:prelims}, we prove the mentioned hardness results in \cref{antler:sec:hardness}. We present structural properties of antlers and how they combine in \cref{sec:def:cuts:antlers}. In \cref{sec:find-fvc} we show how color coding can be used to find a large feedback vertex cut, if one exists. We also prove that, given a large feedback vertex cut, we can shrink it while preserving the antlers in the graph. Our main results are derived in \cref{sec:find-antler}, where we show how color coding can be used to efficiently find antlers when the size of their \antler part is bounded in terms of the size of their \head. We conclude in \cref{sec:conclusion}. 

\section{Preliminaries}\label{antler:sec:prelims} 
\subsection{Graphs and sets}
For a finite set~$A$ we use~$2^A$ to denote the powerset of~$A$. For a function~$f \colon A \to B$, let $\inv{f} \colon B \to 2^A$ denote the \emph{preimage function of $f$}, that is $\inv{f}(a) = \{b \in B \mid f(b) = a\}$.
For any family of sets $X_1,\ldots,X_\ell$ indexed by $\{1,\ldots,\ell\}$ we define the following for all $1 \leq i \leq \ell$: 
\begin{itemize}
\item $X_{<i} := \bigcup_{1 \leq j < i} X_j$, 
\item $X_{>i} := \bigcup_{i < j \leq \ell} X_j$, 
\item $X_{\leq i} := X_i \cup X_{<i}$, and 
\item $X_{\geq i} := X_i \cup X_{>i}$.
\end{itemize}

All graphs considered in this paper are undirected multigraphs, which may have loops. Based on the incidence representation of multigraphs (cf.~\cite[Example 4.9]{FlumG06}) we represent a multigraph $G$ by a vertex set $V(G)$, an edge set $E(G)$, and a function $\iota \colon E(G) \to 2^{V(G)}$ where $\iota(e)$ is the set of one or two vertices incident to $e$ for all $e \in E(G)$. 
In the context of an algorithm with input graph~$G$ we use~$n = |V(G)|$ and~$m = |E(G)|$. 

For a multigraph~$G$ and vertex set~$S \subseteq V(G)$, let~$G[S]$ denote the multigraph induced by~$S$ which consists of the vertex set~$S$, the edge set~$E' = \{e \in E(G) \mid \iota(e) \subseteq S\}$, and the function~$\iota' \colon E' \to 2^S$ defined as~$\iota'(e) = \iota(e)$ for all~$e \in E'$. For a set of vertices and/or edges~$X \subseteq V(G) \cup E(G)$ we define~$G-X$ as the graph obtained from~$G$ after removing all vertices and edges in~$X$, more formally,~$G-X$ is the multigraph consisting of the vertex set~$V(G) \setminus X$, the edge set~$E' = \{e \in E(G) \setminus X \mid \iota(e) \subseteq V(G) \setminus X\}$, and the function~$\iota' \colon E' \to 2^{V(G) \setminus X}$ defined as~$\iota'(e) = \iota(e)$ for all~$e \in E'$. If~$X$ is a singleton set consisting of a single vertex~$v$ or edge~$e$ we may write~$G-v$ or~$G-e$ as shorthand for~$G-\{v\}$ or~$G-\{e\}$.

The open neighborhood of a vertex~$v$ in a multigraph~$G$, denoted~$N_G(v)$, is the set of all vertices~$u \neq v$ for which there is an edge between~$u$ and~$v$, i.e.,~$N_G(v) = \{u \in V(G)\setminus \{v\} \mid \exists e \in E(G) \colon \{u,v\} = \iota(e)\}$. We define~$N_G(S) = \bigcup_{v \in S} N_G(v) \setminus S$ as the open neighborhood of a vertex set~$S \subseteq V(G)$, the closed neighborhood of~$v$ is defined as~$N_G[v] = N_G(v) \cup \{v\}$, and the closed neighborhood of a vertex set~$S \subseteq V(G)$ is defined as~$N_G[S] = N_G(S) \cup S$. For two multigraphs~$G_1, G_2$ on disjoint vertex sets, the \emph{disjoint union} of~$G_1$ and~$G_2$ is the multigraph whose vertex set is~$V(G_1) \cup V(G_2)$ and whose edge set is~$E(G_1) \cup E(G_2)$, with the function~$\iota$ defined in the obvious way.

The \emph{degree}~$\degree_G(v)$ of a vertex~$v$ in a multigraph~$G$ is the number of edge-incidences to~$v$, where a self-loop contributes two edge-incidences. Formally, we define~$\degree_G(v) = |\{e \in E(G) \mid v \in \iota(e)\}| + |\{e \in E(G) \mid \{v\} = \iota(e)\}|$. Note that the \emph{degree-sum theorem},~$\sum_{v \in V(G)} \degree_G(v) = 2|E(G)|$, remains true in multigraphs.
For vertex sets~$A, B \subseteq V(G)$ we define~$E_G(A,B)$ as the set of edges with one endpoint in the vertex set~$A \subseteq V(G)$ and another in the vertex set~$B \subseteq V(G)$, i.e.,~$E_G(A,B) = \{e \in E(G) \mid \iota(e) = \{a,b\}~\text{for some}~a \in A, b \in B\}$. We define~$e_G(A,B) = |E_G(A,B)|$.\index{e@$e_G(..,..)$}

Whenever working with graphs, we may omit the subscript if the graph is clear from the context. To simplify the presentation, in expressions taking one or more vertex sets as parameter such as~$N_G(..)$ and~$E_G(.., ..)$, we sometimes use a subgraph~$H$ of~$G$ as argument as a shorthand for the vertex set~$V(H)$ that is formally needed.

For two multigraphs~$G_1$ and~$G_2$, the graph~$G_1 \cap G_2$ is the multigraph on vertex set~$V' = V(G_1) \cap V(G_2)$, edge set~$E' = E(G_1) \cap E(G_2)$, and function~$\iota' \colon E' \to 2^{V'}$ defined as~$\iota'(e) = \iota(e)$ for all~$e \in E'$.

We assume the multigraphs are stored such that the number of edges between any two vertices can be retrieved and modified in constant time so that ensuring there are at most two edges between any two vertices (and hence~$m \in \Oh(n^2)$) can be done without overhead in the asymptotic runtime of our algorithms.

\subsection{Feedback vertex cuts and antlers}
We now introduce antlers and related structures. A \emph{feedback vertex set} (FVS) in a multigraph~$G$ is a vertex set~$X \subseteq V(G)$ such that~$G-X$ is acyclic (a forest). The feedback vertex number of a graph~$G$, denoted by~$\fvs(G)$, is the minimum size of a FVS in~$G$.
A \emph{feedback vertex cut} (FVC) in a multigraph~$G$ is a pair of disjoint vertex sets~$(C,F)$ such that~$C,F \subseteq V(G)$,~$G[F]$ is a forest, and for each tree~$T$ in~$G[F]$ we have~$e(T,G-(C\cup F)) \leq 1$. The \emph{width} of a FVC~$(C,F)$ is~$|C|$, and~$(C,F)$ is \emph{empty} if~$|C \cup F| = 0$.
\begin{observation} \label{obs:fvc-basics}
 If~$(C,F)$ is a FVC in~$G$, then any cycle in~$G$ containing a vertex from~$F$ also contains a vertex from~$C$. 
\end{observation}
\begin{observation} \label{obs:subfvc}
 If~$(C,F)$ is a FVC in~$G$, then for any~$X \subseteq V(G)$ we have that~$(C\setminus X, F \setminus X)$ is a FVC in~$G-X$. Additionally, for any~$Y \subseteq E(G)$ we have that~$(C,F)$ is a FVC in~$G-Y$.
\end{observation}

We now present one of the main concepts for this work. An \emph{antler} in a multigraph~$G$ is a FVC~$(C,F)$ in~$G$ such that~$|C| \leq \fvs(G[C \cup F])$. The following lemma shows how detecting antlers can reduce the search space for \textsc{Feedback Vertex Set}. 
\begin{lemma} \label{obs:remove-antler}
 If~$(C,F)$ is an antler in~$G$, then~$\fvs(G) = |C| + \fvs(G-(C \cup F))$.
\end{lemma}
\begin{proof}
We establish the equality by proving separate inequalities. We first claim that for any FVS~$X'$ in the graph~$G - (C \cup F)$, the set~$X' \cup C$ is a FVS in~$G$. To prove this, we argue that~$X' \cup C$ intersects each cycle~$\cal{C}$ in~$G$. If~$\cal{C}$ is disjoint from~$F$, then~$\cal{C}$ is intersected by~$C$ or is a cycle in~$G - (C \cup F)$, implying it is intersected by~$X'$. It remains to consider the case that~$\cal{C}$ intersects~$F$. Since~$(C,F)$ is an antler, it is a FVC, so that~$G[F]$ is a forest by definition. Hence~$\mathcal{C}$ intersects some tree~$T$ of the forest~$G[F]$ but also visits a vertex outside~$F$. The cycle~$\cal{C}$ therefore has to enter~$T$ over some edge and later exit~$T$ over a different edge. Since the definition of FVC ensures there is at most one edge that connects a vertex of~$T$ to a vertex outside~$C \cup F$, the cycle~$\mathcal{C}$ either enters or exits tree~$T$ via a vertex of~$C$ and is therefore intersected by~$X' \cup C$. Hence~$X' \cup C$ is indeed a FVS in~$G$. By selecting a minimum FVS~$X'$ in~$G - (C \cup F)$, we obtain~$\fvs(G) \leq |C| + \fvs(G - (C \cup F))$. 

We now prove the converse. Consider an arbitrary FVS~$X$ in~$G$. Since~$X \cap (C \cup F)$ is a FVS of~$G[C \cup F]$, we have~$|X \cap (C \cup F)| \geq \fvs(G[C \cup F]) \geq |C|$, where the second inequality follows from the definition of antler. Since~$X \setminus (C \cup F)$ is trivially a FVS of~$G - (C \cup F)$, we have~$|X \setminus (C \cup F)| \geq \fvs(G - (C \cup F))$. Hence~$|X| \geq |X \cap (C \cup F)| + |X \setminus (C \cup F)| \geq |C| + \fvs(G - (C \cup F))$.
\end{proof}

For a multigraph~$G$ and vertex set~$C \subseteq V(G)$, a \emph{$C$-certificate} is a subgraph~$H$ of~$G$ such that~$C$ is a minimum FVS in~$H$. We say a $C$-certificate has \emph{order}~$z$ if for each component~$H'$ of~$H$ we have~$\fvs(H') = |C \cap V(H')| \leq z$. For an integer~$z\geq0$, a \emph{$z$-antler} in~$G$ is an antler~$(C,F)$ in~$G$ such that~$G[C \cup F]$ contains a $C$-certificate of order~$z$. Note that a $0$-antler has width~$0$. 
In \cref{fig:antler:3-antler} we depicted a feedback vertex cut~$(\head, \antler)$ of width~$5$ where~$G[\head \cup \antler]$ contains a \head-certificate of order~$3$ (marked in blue), hence~$(\head, \antler)$ is a $3$-antler of width~$5$.

\begin{figure}[t]
 \centering
 \includegraphics{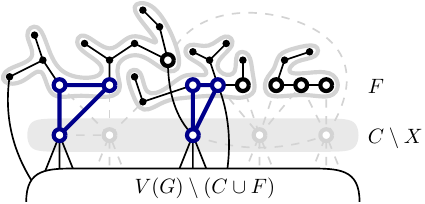}
 \caption{Consider the $3$-antler~$(\head, \antler) = (C,F)$ from \cref{fig:antler:3-antler}. The pair~$(C \setminus X, F)$ remains a $3$-antler after removing a subset~$X \subseteq C$ from~$G$.}
 \label{fig:antler:3-antler-minus-X}
\end{figure}

The following lemma shows a persistence property of antlers when removing some of their vertices. See \cref{fig:antler:3-antler-minus-X} for an illustration.

\begin{lemma} \label{obs:subantler}
 If~$(C,F)$ is a $z$-antler in~$G$ for some~$z\geq0$, then for any~$X \subseteq C$ the pair~$(C\setminus X,F)$ is a $z$-antler in~$G - X$. 
\end{lemma}
\begin{proof}
    Define~$G' := G - X$ and~$C' := C \setminus X$. \cref{obs:subfvc} implies that~$(C', F)$ is a FVC in~$G'$. Next, we establish that~$G'[C' \cup F]$ contains a~$C'$-certificate of order~$z$. Since~$(C,F)$ is a $z$-antler in~$G$, there is a subgraph~$H$ of~$G[C \cup F]$ such that~$C$ is a minimum FVS in~$H$ and each component of~$H$ has feedback vertex number at most~$z$. Now~$H' = H - X$ is a subgraph of~$G'[C' \cup F]$. The subgraph~$H'$ is a $C'$-certificate: the set~$C' = C \setminus X$ is a minimum FVS for~$H'$ since~$C$ is a minimum FVS for~$H$ and~$X \subseteq C$. The order of certificate~$H'$ is not larger than that of~$H$, since each component of~$H'$ is a subgraph of a component of~$H$.

    To prove~$(C',F)$ is a $z$-antler in~$G'$, it therefore remains to prove that~$|C'| \leq \fvs(G'[C' \cup F])$. Since removing a vertex from a graph can decrease its feedback vertex number by at most one, we have~$\fvs(G'[C' \cup F]) = \fvs(G[C \cup F] - X) \geq \fvs(G[C \cup F]) - |X| \geq |C| - |X|$, where the last inequality follows since~$(C,F)$ is an antler. Since~$X \subseteq C$, we have~$|C| - |X| = |C'|$, which combines with the previous inequality to show~$\fvs(G'[C' \cup F']) \geq |C'|$, as desired.
\end{proof}

\section{Hardness results} \label{antler:sec:hardness}

To motivate the use of an FPT algorithm to find antlers, we start by presenting the hardness results mentioned in the introduction. As these results apply to the simplest types of antlers this also forms an introduction to their properties. The hardness results presented in this section apply to the type of antlers as discussed in \cref{antler:sec:intro} consisting of two vertex sets \head and \antler. This type of antler is formally defined in \cref{antler:sec:prelims} as 1-antlers, the simplest type of antler, consisting of the vertex sets~$C$ and~$F$ corresponding to \head and \antler respectively. For convenience we give a self contained definition of a 1-antler below.

\begin{definition}\label{def:1antler}
A \emph{1-antler} in a multigraph~$G$ is a pair of disjoint vertex sets~$(C,F)$ such that:
\begin{enumerate}
	\item $G[F]$ is acyclic,\label[property]{def:1antler:forest}
	\item each tree~$T$ of the forest~$G[F]$ is connected to~$V(G) \setminus (C \cup F)$ by at most one edge, and\label[property]{def:1antler:edgeleaving}
	\item the subgraph~$G[C \cup F]$ contains~$|C|$ vertex-disjoint cycles.\label[property]{def:1antler:cycles}
\end{enumerate}
A 1-antler is called \emph{non-empty} if~$F \cup C \neq \emptyset$.
\end{definition}

Observe that the combination of \cref{def:1antler:forest,def:1antler:edgeleaving} is equivalent to stating that~$(C,F)$ is a FVC, and \cref{def:1antler:cycles} describes the existence of a $C$-certificate of order~$1$ in~$G[C \cup F]$. We will use the following easily verified consequence of this definition.

\begin{observation} \label{obs:1antler:cycleWithF}
If~$(C,F)$ is a 1-antler in an multigraph~$G$, then for each vertex~$c \in C$ the subgraph~$G[\{c\} \cup F]$ contains a cycle.
\end{observation}

\subsection{\NP-hardness of finding 1-antlers} \label{sec:np-hard}

For the hardness results in this section we use a reduction from \cnfsat given in an earlier work by the authors.

\begin{lemma}[{\cite[Lemma 13 for~$H = K_3$]{DonkersJ21}}] \label{lemma:reduction:disjointcycles}
There is a polynomial-time algorithm that, given a CNF formula~$\Phi$, outputs a graph~$G$ and a collection~$\mathcal{H} = \{H_1, \ldots, H_\ell\}$ of vertex-disjoint cycles in~$G$, such that~$\Phi$ is satisfiable if and only if~$G$ has a feedback vertex set of size~$\ell$.
\end{lemma}

We now show that finding non-empty 1-antlers is \NPhard, even when the input graph is simple.

\begin{theorem}
Assuming \PNEQNP, there is no polynomial-time algorithm that, given a simple graph~$G$, outputs a non-empty 1-antler in~$G$ or concludes that no non-empty 1-antler exists.
\end{theorem}
\begin{proof}
We need the following simple claim.

\begin{claim}
If~$G$ contains a packing of~$\ell \geq 1$ vertex-disjoint cycles and a feedback vertex set of size~$\ell$, then~$G$ admits a non-empty 1-antler~$(C,F)$.
\end{claim}
\begin{claimproof}
Let~$C$ be a feedback vertex set in~$G$ of size~$\ell$ and let~$F := V(G) \setminus C$.
\end{claimproof}

Now suppose there is a polynomial-time algorithm to find a non-empty 1-antler decomposition, if one exists. We use it to solve \cnfsat in polynomial time. Given an input formula~$\Phi$ for \cnfsat, use \cref{lemma:reduction:disjointcycles} to produce in polynomial time a graph~$G$ and a packing~$\mathcal{H}$ of~$\ell$ vertex-disjoint cycles in~$G$, such that~$\Phi$ is satisfiable if and only if~$\fvs(G) = \ell$. The following recursive polynomial-time algorithm correctly tests, given a graph~$G$ and a packing~$\mathcal{H}$ of some~$\ell \geq 0$ vertex-disjoint cycles in~$G$, whether~$\fvs(G) = \ell$.

\begin{enumerate}
	\item If~$\ell = 0$, then output \textsc{yes} if and only if~$G$ is acyclic.
	\item If~$\ell > 0$, run the hypothetical algorithm to find a non-empty 1-antler~$(C,F)$. 
	\begin{enumerate}
		\item If a non-empty 1-antler~$(C,F)$ is returned, then let~$\mathcal{H'}$ consist of those cycles in the packing not intersected by~$C$, and let~$\ell' := |\mathcal{H'}|$. Return the result of recursively running the algorithm on~$G' := G - (C \cup F)$ and~$\mathcal{H'}$ to test whether~$G'$ has a feedback vertex set of size~$|\mathcal{H'}|$.
		\item Otherwise, return \textsc{no}.
	\end{enumerate}
\end{enumerate}

The claim shows that the algorithm is correct when it returns \textsc{no}. \Cref{obs:remove-antler} shows that if we recurse, we have~$\fvs(G) = \ell$ if and only if~$\fvs(G - (C \cup H)) = \ell'$; hence the result of the recursion is the correct output. Since the number of vertices in the graph reduces by at least one in each iteration, the overall running time is polynomial assuming the hypothetical algorithm to compute a non-empty 1-antler. Hence using \cref{lemma:reduction:disjointcycles} we can decide \cnfsat in polynomial time.
\end{proof}

Having established the NP-hardness of detecting non-empty 1-antlers, we consider the standard linear-programming relaxation of \textsc{Feedback Vertex Set} on a multigraph~$G$. It is a linear program that has one variable~$x_v$ for each vertex~$v \in V(G)$ with~$0 \leq x_v \leq 1$, along with a constraint for each cycle~$C \subseteq V(G)$ stating that~$\sum _{v \in C} x_c \geq 1$. We show that determining whether the cost of a solution to the linear-programming relaxation of \scFVS on~$G$ is equal to the feedback vertex number of~$G$ is \NPcomplete. 

\begin{corollary} \label{cor:lp-relax}
For a simple graph~$G$, determining whether~$\fvs(G) = \fvsLP(G)$ is \NPcomplete. Here~$\fvsLP(G)$ denotes the minimum cost of a solution to the linear-programming relaxation of \scFVS on~$G$.
\end{corollary}
\begin{proof}
Membership in~\NP is trivial; we prove hardness by giving a polynomial-time reduction from \cnfsat to the problem of testing whether~$\fvs(G) = \fvsLP(G)$. The reduction proceeds as follows. 
Given an input~$\Phi$ for \cnfsat, use \cref{lemma:reduction:disjointcycles} to produce in polynomial time a graph~$G$ and packing~$\mathcal{H}$ of~$\ell$ vertex-disjoint cycles in~$G$, such that~$\Phi$ is satisfiable if and only if~$\fvs(G) = \ell$. 

Compute the cost~$c$ of an optimal solution to the linear-programming relaxation of \scFVS on~$G$, using the ellipsoid method. By the properties of a relaxation, if~$c > \ell$ then~$\fvs(G) > \ell$, and hence~$\Phi$ is unsatisfiable; we output a trivial \textsc{no}-instance as the result of the reduction. If~$c \leq \ell$, then the existence of~$\ell$ vertex-disjoint cycles in~$G$ implies that~$c = \ell$. We output~$G$ as the result of the reduction, which guarantees that the feedback vertex number of~$G$ is equal to~$\ell = c = \fvsLP(G)$ if and only if~$\Phi$ is satisfiable.
\end{proof}

\subsection{\Wone-hardness of finding bounded-width 1-antlers} \label{sec:w1-hard}

We consider the following parameterized problem.

\defparproblem{\scBWoneAntlerDetection}
{A multigraph~$G$ and an integer~$k$.}
{$k$.}
{Does~$G$ admit a non-empty 1-antler~$(C,F)$ with~$|C| \leq k$?}

We prove that \scBWoneAntlerDetection is \Wonehard by a reduction from \scMultColoredClique, which is defined as follows.

\defparproblem{\scMultColoredClique}
{A simple graph~$G$, an integer~$k$, and a partition of~$V(G)$ into sets~$V_1, \ldots, V_k$.}
{$k$.}
{Is there a clique~$S$ in~$G$ such that for each~$1 \leq i \leq k$ we have~$|S \cap V_i| = 1$?}

The sets~$V_i$ are referred to as \emph{color classes}, and a solution clique~$S$ is called a \emph{multicolored clique}. It is well-known that \scMultColoredClique is \Wonehard (cf.~\cite[Theorem~13.25]{CyganFKLMPPS15}). Our reduction is inspired by the \Wone-hardness of detecting a Hall set~\cite[Exercise 13.28]{CyganFKLMPPS15}.

\begin{theorem} \label{thm:1antler:w1hard}
\scBWoneAntlerDetection is \Wonehard.
\end{theorem}
\begin{proof}
We give a parameterized reduction from the \scMultColoredClique problem. By inserting isolated vertices if needed, we may assume without loss of generality that for the input instance~$(G,k,V_1, \ldots, V_k)$ we have~$|V_1| = |V_2| = \ldots = |V_k| = n$ for some~$n \geq k(k-1) + 4$. For each~$i \in [k]$, fix an arbitrary labeling of the vertices in~$V_i$ as~$v_{i,j}$ for~$j \in [n]$. Given this instance, we construct an input~$(G', k')$ for \scBWoneAntlerDetection as follows.
\begin{enumerate}
	\item For each~$i \in [k]$, for each~$j \in [n]$, create a set~$U_{i,j} = \{u_{i,j,\ell} \mid \ell \in [k] \setminus \{i\}\}$ of vertices in~$G'$ to represent~$v_{i,j}$. 
    Intuitively, vertex~$u_{i,j,\ell}$ is used to check the presence of an edge from the~$j$th vertex of the~$i$th color class to whichever vertex from the $\ell$th color class is chosen in the solution clique.
	\item Define~$\mathcal{U} := \bigcup _{i \in [k]} \bigcup _{j \in [n]} U_{i,j}$. Insert (single) edges to turn~$\mathcal{U}$ into a clique in~$G'$.
	\item \label[step]{w1hard:w} For each edge~$e$ in~$G$ between vertices of different color classes, let~$e = \{v_{i,j}, v_{i',j'}\}$ with~$i < i'$, and insert two vertices into~$G'$ to represent~$e$:
	\begin{itemize}
		\item Insert a vertex~$w_e$, add an edge from~$w_e$ to each vertex in~$U_{i,j} \cup U_{i', j'}$, and then add a second edge between~$w_e$ and~$u_{i,j,i'}$.
		\item Insert a vertex~$w_{e'}$, add an edge from~$w_{e'}$ to each vertex in~$U_{i,j} \cup U_{i', j'}$, and then add a second edge between~$w_{e'}$ and~$u_{i',j',i}$.
	\end{itemize}
	Let~$W$ denote the set of vertices of the form~$w_e, w_{e'}$ inserted to represent an edge of~$G$. Observe that~$W$ is an independent set in~$G'$.
	\item Finally, insert a vertex~$u^*$ into~$G'$. Add a single edge from~$u^*$ to all other vertices of~$G'$, to make~$u^*$ into a universal vertex.
\end{enumerate}

This concludes the construction of~$G'$. Note that~$G'$ contains double-edges, but no self-loops. We set~$k' := k(k-1)$, which is appropriately bounded for a parameterized reduction. It is easy to see that the reduction can be performed in polynomial time. It remains to show that~$G$ has a multicolored $k$-clique if and only if~$G'$ has a non-empty 1-antler of width at most~$k'$. To illustrate the intended behavior of the reduction, we first prove the forward implication.

\begin{claim} \label{claim:w1hard:forward}
If~$G$ has a multicolored clique of size~$k$, then~$G'$ has a non-empty 1-antler of width~$k'$.
\end{claim}
\begin{claimproof}
Suppose~$S$ is a multicolored clique of size~$k$ in~$G$. Choose indices~$j_1, \ldots, j_k$ such that~$S \cap V_i = \{v_{i,j_i}\}$ for all~$i \in [k]$. Define a 1-antler~$(C,F)$ in~$G'$ as follows:
\begin{align*}
	C &= \bigcup _{i \in [k]} U_{i, j_i},\\
	F &= \{w_e, w_{e'} \mid e \text{~is an edge in~$G$ between distinct vertices of~$S$}\}.
\end{align*}
Since each set~$U_{i,j_i}$ has size~$k-1$, it follows that~$|C| = k(k-1) = k'$. Since~$F \subseteq W$ is an independent set in~$G'$, it also follows that~$G'[F]$ is acyclic. Each tree~$T$ in the forest~$G'[F]$ consists of a single vertex~$w_e$ or~$w_{e'}$. By construction, there is exactly one edge between~$T$ and~$V(G') \setminus (C \cup F)$; this is the edge to the universal vertex~$u^*$. It remains to verify that~$G'[C \cup F]$ contains~$|C|$ vertex-disjoint cycles, each containing exactly one vertex of~$C$. Consider an arbitrary vertex~$u_{i, j_i, \ell}$ in~$C$; we show we can assign it a cycle in~$G'[C \cup F]$ so that all assigned cycles are vertex-disjoint. Since~$S$ is a clique, there is an edge~$e$ in~$G$ between~$v_{i, j_i}$ and~$v_{\ell, j_\ell}$, and the corresponding vertices~$w_e, w_{e'}$ are in~$F$. If~$i < \ell$, then~$w_e \in F$ and there are two edges between~$u_{i, j_i, \ell}$ and~$w_e$, forming a cycle on two vertices. If~$i > \ell$, then there is a cycle on two vertices~$u_{i, j_i, \ell}$ and~$w_{e'}$. Since for any vertex of the form~$w_e$ or~$w_{e'}$ there is a unique vertex of~$C$ that it has a double-edge to, the resulting cycles are indeed vertex-disjoint. This proves that~$(C,F)$ is a 1-antler of width~$k'$.
\end{claimproof}

Before proving reverse implication, we establish some structural claims about the structure of 1-antlers in~$G'$.

\begin{claim} \label{claim:w1hard:antlers}
If~$(C,F)$ is a non-empty 1-antler in~$G'$ with~$|C| \leq k'$, then all of the following conditions hold: 
\begin{enumerate}
	\item $(\mathcal{U} \cup \{u^*\}) \cap F = \emptyset$.\label[condition]{w1:antler:u}
	\item $u^* \notin C \cup F$.\label[condition]{w1:antler:ustar} 
	\item $W \cap C = \emptyset$.\label[condition]{w1:antler:w}
	\item $C \subseteq \mathcal{U}$,~$F \subseteq W$, and each tree of the forest~$G'[F]$ consists of a single vertex.\label[condition]{w1:antler:subsets}
	\item For each vertex~$w \in F$ we have~$N_{G'}(w) \cap \mathcal{U} \subseteq C$.\label[condition]{w1:antler:neighborhoodinC}
	\item $F \neq \emptyset$.\label[condition]{w1:antler:fnonempty}
\end{enumerate}
\end{claim}
\begin{claimproof}
\Cref{w1:antler:u}:
Assume for a contradiction that there is a vertex~$u \in \mathcal{U} \cup \{u^*\}$ such that~$u \in F$. Since~$G'[F]$ is a forest by \cref{def:1antler:forest} of \cref{def:1antler}, while~$\mathcal{U} \cup \{u^*\}$ is a clique in~$G'$, it follows that~$|F \cap (\mathcal{U} \cup \{u^*\})| \leq 2$. By \cref{def:1antler:edgeleaving}, for a vertex in~$F$, there is at most one of its neighbors that belongs to neither~$F$ nor~$C$. Since~$|\mathcal{U} \cup \{u^*\}| \geq n \geq k(k-1) + 4$, and~$u \in F$ is adjacent to all other vertices of~$\mathcal{U} \cup \{u^*\}$ since that set forms a clique, the fact that~$|F \cap (\mathcal{U} \cup \{u^*\})| \leq 2$ implies that~$|C \cap (\mathcal{U} \cup \{u^*\})| \geq (|\mathcal{U}| + 1) - 2 - 1 \geq k(k-1) + 5 - 3 > k'$. So~$|C| > k'$, which contradicts that~$(C,F)$ is a 1-antler with~$|C| \leq k'$.

\Cref{w1:antler:ustar}:
The preceding item shows that~$u^* \notin F$. To prove the claim we show that~$u^* \notin C$. Assume for a contradiction that~$u^* \in C$. By \cref{obs:1antler:cycleWithF}, the graph~$G'[\{u^*\} \cup F]$ contains a cycle. Since~$W$ is an independent set in~$G'$ and~$u^*$ is not incident on any double-edges, the graph~$G'[\{u^*\} \cup W]$ is acyclic. Hence to get a cycle in~$G'[\{u^*\} \cup F]$, the set~$F$ contains at least one vertex that is not in~$W$ and not~$u^*$; hence this vertex belongs to~$\mathcal{U}$. So~$\mathcal{U} \cap F \neq \emptyset$; but this contradicts \cref{w1:antler:u}.

\Cref{w1:antler:w}:
Assume for a contradiction that~$w \in W \cap C$. Again by \cref{obs:1antler:cycleWithF}, there is a cycle in~$G'[\{w\} \cup F]$, and since~$G'$ does not have any self-loops this implies~$N_{G'}(w) \cap F \neq \emptyset$. But by construction of~$G'$ we have~$N_{G'}(w) \subseteq \mathcal{U} \cup \{u^*\}$, so~$F$ contains a vertex of either~$\mathcal{U}$ or~$u^*$. But this contradicts either \cref{w1:antler:u} or \cref{w1:antler:ustar}.

\Cref{w1:antler:subsets}:
Since the sets~$\mathcal{U}, W, \{u^*\}$ form a partition of~$V(G')$, the preceding claims imply~$C \subseteq \mathcal{U}$ and~$F \subseteq W$. Since~$W$ is an independent set in~$G'$, this implies that each tree~$T$ of the forest~$G'[W]$ consists of a single vertex.

\Cref{w1:antler:neighborhoodinC}: Consider a vertex~$w \in F$, which by itself forms a tree~$T$ in~$G'[F]$. Since~$u^* \notin C \cup F$, the edge between~$T$ and~$u^*$ is the unique edge connecting~$T$ to a vertex of~$V(G') \setminus (C \cup F)$, and therefore all neighbors of~$T$ other than~$u^*$ belong to~$C \cup F$. Since a vertex~$w \in W$ has~$N_{G'}(w) \subseteq \mathcal{U} \cup \{u^*\}$, it follows that~$N_{G'}(w) \cap \mathcal{U} \subseteq C$.

\Cref{w1:antler:fnonempty}: By the assumption that~$(C,F)$ is non-empty, we have~$C \cup F \neq \emptyset$. This implies that~$F \neq \emptyset$: if~$C$ would contain a vertex~$c$ while~$F = \emptyset$, then by \cref{obs:1antler:cycleWithF} the graph~$G'[\{c\} \cup F] = G'[\{c\}]$ would contain a cycle, which is not the case since~$G'$ has no self-loops. Hence~$F \neq \emptyset$.
\end{claimproof}

With these structural insights, we can prove the remaining implication.

\begin{claim} \label{claim:w1hard:reverse}
If~$G'$ has a non-empty 1-antler~$(C,F)$ with~$|C| \leq k'$, then~$G$ has a multicolored clique of size~$k$.
\end{claim}
\begin{claimproof}
Let~$(C,F)$ be a non-empty 1-antler in~$G'$ with~$|C| \leq k'$. By \cref{claim:w1hard:antlers} we have~$C \subseteq \mathcal{U}$, while~$F \subseteq W$ and~$F \neq \emptyset$. Consider a fixed vertex~$w \in F$. Since~$F \subseteq W$, vertex~$w$ is of the form~$w_e$ or~$w_{e'}$ constructed in \cref{w1hard:w} to represent some edge~$e$ of~$G$. Choose~$i^* \in [k], j^* \in [n]$ such that~$v_{i^*,j^*} \in V_{i^*}$ is an endpoint of edge~$e$ in~$G$. By construction we have~$U_{i^*,j^*} \subseteq N_{G'}(w)$ and therefore \cref{claim:w1hard:antlers}\eqref{w1:antler:neighborhoodinC} implies~$U_{i^*,j^*} \subseteq C$. 

Consider an arbitrary~$\ell \in [k] \setminus \{i^*\}$. Then~$u_{i^*,j^*,\ell} \in U_{i^*,j^*} \subseteq C$, so by \cref{obs:1antler:cycleWithF} the graph~$G'[\{u_{i^*,j^*,\ell}\} \cup F]$ contains a cycle. Since~$G'[F]$ is an independent set and~$G'$ has no self-loops, this cycle consists of two vertices joined by a double-edge. By construction of~$G'$, such a cycle involving~$u_{i^*,j^*,\ell}$ exists only through vertices~$w_e$ or~$w_{e'}$ where~$e$ is an edge of~$G$ connecting~$v_{i^*,j^*}$ to a neighbor in class~$V_{\ell}$. Consequently,~$F$ contains a vertex~$w'$ that represents such an edge~$e$. Let~$v_{\ell, j_\ell}$ denote the other endpoint of~$e$. Then~$N_{G'}(w') \supseteq U_{i^*,j^*} \cup U_{\ell, j_\ell}$, and by \cref{claim:w1hard:antlers}\eqref{w1:antler:neighborhoodinC} we therefore have~$U_{\ell, j_\ell} \subseteq C$. 

Applying the previous argument for all~$\ell \in [k] \setminus \{i^*\}$, together with the fact that~$U_{i^*,j^*} \subseteq C$, we find that for each~$i \in [k]$ there exists a value~$j_i$ such that~$U_{i, j_i} \subseteq C$. Since~$|C| \leq k(k-1)$ while each such set~$U_{i, j_i}$ has size~$k-1$, it follows that the choice of~$j_i$ is uniquely determined for each~$i \in [k]$, and that there are no other vertices in~$C$. To complete the proof, we argue that the set~$S = \{ v_{i, j_i} \mid i \in [k] \}$ is a clique in~$G$. 

Consider an arbitrary pair of distinct vertices~$v_{i, j_i}, v_{i', j_{i'}}$ in~$S$, and choose the indices such that~$i < i'$. We argue that~$G$ contains an edge~$e$ between these vertices, as follows. Since~$u_{i, j_i, i'} \in U_{i, j_i} \subseteq C$, by \cref{obs:1antler:cycleWithF} the graph~$G'[\{u_{i, j_i, i'}\} \cup F]$ contains a cycle. As argued above, the construction of~$G'$ and the fact that~$F \subseteq W$ ensure that this cycle consists of~$u_{i, j_i, i'}$ joined to a vertex in~$F$ by a double-edge. By \cref{w1hard:w} and the fact that~$i < i'$, this vertex is of the form~$w_e$ for an edge~$e$ in~$G$ connecting~$v_{i, j_i}$ to a vertex~$v_{i',j'}$ in~$V_{i'}$. By construction of~$G'$ we have~$U_{i', j'} \subseteq N_{G'}(w_e)$, and then~$w_e \in F$ implies by \cref{claim:w1hard:antlers}\eqref{w1:antler:neighborhoodinC} that~$U_{i', j'} \subseteq C$. Since we argued above that for index~$i'$ there is a unique choice~$j_{i'}$ with~$U_{i', j_{i'}} \subseteq C$, we must have~$j' = j_{i'}$. Hence the vertex~$w_e$ contained in~$F$ represents the edge of~$G$ between~$v_{i, j_i}$ and~$v_{i', j_{i'}}$ in~$G$, which proves in particular that the edge exists. As the choice of vertices was arbitrary, this shows that~$S$ is a clique in~$G$. As it contains exactly one vertex from each color class, graph~$G$ has a multicolored clique of size~$k$.
\end{claimproof}

\Cref{claim:w1hard:forward,claim:w1hard:reverse} show that the instance~$(G,k)$ of \scMultColoredClique is equivalent to the instance~$(G',k')$ of \scBWoneAntlerDetection. This concludes the proof of \cref{thm:1antler:w1hard}.
\end{proof}

Observe that the proof of \cref{thm:1antler:w1hard} shows that the variant of \scBWoneAntlerDetection where we ask for the existence of a 1-antler of width \emph{exactly}~$k$, is also \Wonehard. In the remainder of this paper we develop algorithms which bypass these hardness proofs because they may find antlers whose head is larger than~$k$ in graphs that have an antler whose head has size exactly~$k$. Note that for the purpose of preprocessing, finding an antler with a larger head is advantageous since it means identifying a larger part of an optimal solution.

\section{Structural properties of antlers} \label{sec:def:cuts:antlers}
In this section we give a number of structural properties of antlers.
While antlers may intersect in non-trivial ways, the following proposition relates the sizes of the cross-intersections.

\begin{proposition}\label{prop:intersection_antlers}
 If~$(C_1,F_1)$ and~$(C_2,F_2)$ are antlers in~$G$, then~$|C_1 \cap F_2| = |C_2 \cap F_1|$.
\end{proposition}
\begin{proof}
 We prove the equality by showing both quantities are equal to~$\fvs(G[F_1 \cup F_2])$. We begin by showing~$\fvs(G[F_1 \cup F_2]) = |C_1 \cap F_2|$.
 
 First we show~$\fvs(G[F_1 \cup F_2]) \geq |C_1 \cap F_2|$ by showing~$(C_1 \cap F_2, F_1)$ is an antler in~$G[F_1 \cup F_2]$. Clearly~$(C_1,F_1)$ is an antler in~$G[F_1 \cup F_2 \cup C_1]$, so then by \cref{obs:subantler}~$(C_1 \cap F_2,F_1)$ is an antler in~$G[F_1 \cup F_2 \cup C_1] - (C_1 \setminus F_2) = G[F_1 \cup F_2]$.
 
 Second we show~$\fvs(G[F_1 \cup F_2]) \leq |C_1 \cap F_2|$ by showing~$G[F_1 \cup F_2] - (C_1 \cap F_2)$ is acyclic. Note that~$G[F_1 \cup F_2] - (C_1 \cap F_2) = G[F_1 \cup F_2] - C_1$. 
 Suppose~$G[F_1 \cup F_2] - C_1$ contains a cycle. We know this cycle does not contain a vertex from~$C_1$, however it does contain at least one vertex from~$F_1$ since otherwise this cycle exists in~$G[F_2]$ which is a forest. 
 We know from \cref{obs:fvc-basics} that any cycle in~$G$ containing a vertex from~$F_1$ also contains a vertex from~$C_1$. Contradiction. The proof for~$\fvs(G[F_1 \cup F_2]) = |C_2 \cap F_1|$ is symmetric. It follows that~$|C_1 \cap F_2| = \fvs(G[F_1 \cup F_2]) = |C_2 \cup F_1|$.
\end{proof}

\Cref{lem:antler_diff} shows that what remains of a $z$-antler~$(C_1,F_1)$ when removing a different antler~$(C_2,F_2)$, again forms a smaller $z$-antler. We will rely on this lemma repeatedly to ensure that after having found and removed a proper subset of a width-$k$ $z$-antler, the remainder of that antler persists as a $z$-antler to be found later.

\begin{lemma}\label{lem:antler_diff}
 For any integer~$z\geq0$, if a multigraph~$G$ has a $z$-antler~$(C_1,F_1)$ and another antler~$(C_2,F_2)$, then~$(C_1 \setminus (C_2 \cup F_2), F_1 \setminus (C_2 \cup F_2))$ is a $z$-antler in~$G-(C_2 \cup F_2)$.
\end{lemma}

Before we prove \cref{lem:antler_diff}, we prove a weaker claim:

\begin{proposition}\label{prop:intersecting_antlers}
 If~$(C_1,F_1)$ and~$(C_2,F_2)$ are antlers in~$G$, then~$(C_1 \setminus (C_2 \cup F_2), F_1 \setminus (C_2 \cup F_2))$ is an antler in~$G - (C_2 \cup F_2)$.
\end{proposition}
\begin{proof}
\begin{figure}[t]
 \centering
 \includegraphics{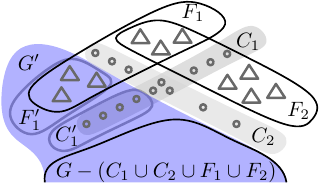}
 \caption{A diagram of the different vertex sets used in the proof of \cref{prop:intersecting_antlers}. The triangles represent induced trees.}
 \label{fig:antler:intersecting-antlers}
\end{figure}
For brevity let~$C_1' := C_1 \setminus (C_2 \cup F_2)$ and~$F_1' := F_1 \setminus (C_2 \cup F_2)$ and~$G' := G - (C_2 \cup F_2)$ (see \cref{fig:antler:intersecting-antlers}). 
First note that~$(C_1', F_1')$ is a FVC in~$G'$ by \cref{obs:subfvc}. We proceed to show that~$\fvs(G'[C_1' \cup F_1']) \geq |C_1'|$. By \cref{obs:subantler}, the pair~$(\emptyset,F_2)$ is an antler in~$G-C_2$, so then by \cref{obs:subfvc} we have~$(\emptyset, F_2 \cap (C_1 \cup F_1))$ is a FVC in~$G[C_1 \cup F_1] - C_2$. Since a FVC of width~$0$ is an antler we can apply \cref{obs:remove-antler} and obtain~$\fvs(G[C_1 \cup F_1] - C_2) = \fvs(G[C_1 \cup F_1] - (C_2 \cup F_2)) = \fvs(G'[C_1' \cup F_1'])$. We derive
\begin{align*}
\fvs(G'[C_1' \cup F_1'])
&=    \fvs(G[C_1 \cup F_1] - C_2)\\
&\geq \fvs(G[C_1 \cup F_1]) - |C_2 \cap (C_1 \cup F_1)|\\
&=    |C_1| - |(C_2 \cap C_1) \cup (C_2 \cap F_1)|  &\text{Since $(C_1, F_1)$ is an antler}\\ 
&=    |C_1| - |C_2 \cap C_1| - |C_2 \cap F_1| &\text{Since $C_1 \cap F_1 = \emptyset$}\\
&=    |C_1| - |C_2 \cap C_1| - |C_1 \cap F_2| &\text{By \cref{prop:intersection_antlers}}\\
&=    |C_1| - |(C_2 \cap C_1) \cup (C_1 \cap F_2)| &\text{Since $C_2 \cap F_2 = \emptyset$}\\
&=    |C_1 \setminus (C_2 \cup F_2)| = |C_1'|.
\tag*{\qedhere}
\end{align*}
\end{proof}

We can now prove \cref{lem:antler_diff}.
\begin{proof}
 For brevity let~$C_1' := C_1 \setminus (C_2 \cup F_2)$ and~$F_1' := F_1 \setminus (C_2 \cup F_2)$ and~$G' := G - (C_2 \cup F_2)$.
 By \cref{prop:intersecting_antlers} we know~$(C_1',F_1')$ is an antler in~$G'$, so it remains to show that~$G'[C_1' \cup F_1']$ contains a $C_1'$-certificate of order~$z$. Since~$(C_1,F_1)$ is a $z$-antler in~$G$, we have that~$G[C_1 \cup F_1]$ contains a $C_1$-certificate of order~$z$. Let~$H$ denote this $C_1$-certificate. Let~$\overline{V_H}$ be the set of all vertices in~$G'[C_1' \cup F_1']$ that are not in~$H$; similarly let~$\overline{E_H}$ be the set of edges in~$G'[C_1' \cup F_1']$ that are not in~$H$. 
 
 We have~$C_1 \subseteq V(H)$ since~$C_1$ is a minimum FVS in~$H$ by definition of certificate. Let~$G''$ be the multigraph obtained from~$G$ by removing the vertices of~$\overline{V_H}$ and edges of~$\overline{E_H}$. Then~$(C_1,F_1 \setminus \overline{V_H})$ is a $z$-antler in~$G''$ since it is a FVC by \cref{obs:subfvc} and~$G''[C_1 \cup (F_1 \setminus \overline{V_H})]$ contains a $C_1$-certificate of order~$z$ as~$H$ is a subgraph of~$G''$. Note that~$(C_2,F_2)$ is also an antler in~$G''$ since~$\overline{V_H} \cup \overline{E_H}$ does not contain vertices or edges from~$G[C_2 \cup F_2]$. It follows that~$(C_1',F_1' \setminus \overline{V_H})$ is an antler in~$G''$ by \cref{prop:intersecting_antlers}, so~$G''[C_1' \cup (F_1' \setminus \overline{V_H})]$ is a $C_1'$-certificate in~$G''$. Clearly this is a $C_1'$-certificate of order~$z$ since~$G''[C_1' \cup (F_1' \setminus \overline{V_H})]$ is a subgraph of~$H$. Since~$G''[C_1' \cup (F_1' \setminus \overline{V_H})]$ is a subgraph of~$G'[C_1' \cup (F_1' \setminus \overline{V_H})]$ it follows that~$G'[C_1' \cup (F_1' \setminus \overline{V_H})]$ contains a $C_1'$-certificate of order~$z$.
\end{proof}

\Cref{lem:antler_combine} shows that we can consider consecutive removal of two $z$-antlers as the removal of a single $z$-antler.

\begin{lemma}\label{lem:antler_combine}
 For any integer~$z\geq0$, if a multigraph~$G$ has a $z$-antler~$(C_1,F_1)$ and~$G-(C_1 \cup F_1)$ has a $z$-antler~$(C_2,F_2)$ then~$(C_1 \cup C_2, F_1 \cup F_2)$ is a $z$-antler in~$G$.
\end{lemma}
\begin{proof}
 Since~$(C_1,F_1)$ is a $z$-antler in~$G$ we know~$G[C_1 \cup F_1]$ contains a $C_1$-certificate~$H_1$ of order~$z$, similarly~$(G-(C_1 \cup F_1))[C_2 \cup F_2]$ contains a $C_2$-certificate~$H_2$ of order~$z$. Since~$H_1$ and~$H_2$ are vertex-disjoint, the disjoint union~$H^* = H_1 \cup H_2$ of these certificates forms a $(C_1 \cup C_2)$-certificate of order~$z$ in~$G[C_1 \cup C_2 \cup F_1 \cup F_2]$: the definition of certificate ensures that~$C_i$ is a minimum FVS of~$H_i$ for each~$i \in [2]$, so that~$C_1 \cup C_2$ is a minimum FVS in~$H^*$. It remains to show that~$(C_1 \cup C_2, F_1 \cup F_2)$ is a FVC in~$G$.
 
 First we show~$G[F_1 \cup F_2]$ is acyclic. Suppose for contradiction that~$G[F_1 \cup F_2]$ contains a cycle~$\mathcal{C}$. Since~$(C_1,F_1)$ is a FVC in~$G$, any cycle containing a vertex from~$F_1$ also contains a vertex from~$C_1$, hence~$\mathcal{C}$ does not contain vertices from~$F_1$. Therefore~$\mathcal{C}$ can only contain vertices from~$F_2$. This is a contradiction with the fact that~$G[F_2]$ is acyclic.
 
 Finally we show that for each tree~$T$ in~$G[F_1 \cup F_2]$ we have~$e(T,G-(C_1 \cup C_2 \cup F_1 \cup F_2)) \leq 1$. If~$V(T) \subseteq F_2$ this follows directly from the fact that~$(C_2,F_2)$ is a FVC in~$G-(C_1 \cup F_1)$. Similarly if~$V(T) \subseteq F_1$ this follows directly from the fact that~$(C_1,F_1)$ is a FVC in~$G$. So suppose~$T$ is a tree that contains vertices from both~$F_1$ and~$F_2$. 
 
 Since~$T$ is connected, each tree~$T'$ in~$T[F_1  \cap V(T)]$ has an edge to a vertex in a tree in~$T[F_2  \cap V(T)]$, which must be the unique edge that connects~$T'$ to a vertex outside~$C_1 \cup F_1$ with respect to the FVC~$(C_1, F_1)$. Hence no tree in~$T[F_1 \cap V(T)]$ contains a neighbor of~$V(G-(C_1 \cup C_2 \cup F_1 \cup F_2))$, so~$e(F_1 \cap V(T), G-(C_1 \cup C_2 \cup F_1 \cup F_2)) = 0$. 
 
 To complete the proof we show~$e(F_2 \cap V(T), G-(C_1 \cup C_2 \cup F_1 \cup F_2)) \leq 1$. Recall that each tree in~$G[F_2]$ has at most~$1$ edge to~$G-(C_1 \cup C_2 \cup F_1 \cup F_2)$, so it suffices to show that~$T[F_2 \cap V(T)]$ is connected. Suppose~$T[F_2 \cap V(T)]$ is not connected, then let~$u,v \in F_2$ be vertices from different components of~$T[F_2 \cap V(T)]$. Since~$T$ is connected, there is a path from~$u$ to~$v$ inside~$T$. This path must use a vertex~$w \in V(T-F_2) \subseteq F_1$. Let~$T''$ denote the tree in~$T[F_1 \cap V(T)]$ that contains~$w$. Since~$(C_1,F_1)$ is a FVC in~$G$ we have that~$e(T'', F_2) \leq e(T'',G-(C_1 \cup F_1)) \leq 1$ hence no vertex in~$T''$ can be part of a path from~$u$ to~$v$ in~$T$. This contradicts our choice of~$T''$ and completes the proof.
\end{proof}

The last structural property of antlers, given in \cref{col:few-fvcs}, derives a bound on the number of trees of a forest~$G[F]$ needed to witness that~$C$ is an optimal FVS of~$G[C \cup F]$. \Cref{col:few-fvcs} is a corollary to the following lemma.

\begin{lemma} \label{lem:small-cert}
 If a multigraph~$G$ contains a $C$-certificate~$H$ of order~$z \geq 0$ for some~$C\subseteq V(G)$, then~$H$ contains a $C$-certificate~$\hat{H}$ of order~$z$ such that~$\hat{H}-C$ has at most~$\frac{|C|}{2}(z^2+4z-1)$ trees.
\end{lemma}
\begin{proof}
 Consider a tree~$T$ in~$H-C$, we show that~$\fvs(H-V(T)) = \fvs(H)$ if
 \begin{enumerate}
  \item \label[condition]{item:vflower} for all~$v \in C$ such that~$H[V(T) \cup \{v\}]$ has a cycle,~$H-V(T)$ contains~$z$ cycles whose vertex sets only intersect in~$v$, and
  \item \label[condition]{item:pumpkin} for all~$\{u,v\} \in \binom{N_H(T)}{2}$ there are at least~$z+1$ other trees in~$H-C$ which are simultaneously adjacent to~$u$ and~$v$.
 \end{enumerate}
 
 Consider the component~$H'$ of~$H$ containing~$T$. It suffices to show that~$\fvs(H'-V(T)) = \fvs(H')$. Clearly~$\fvs(H' - V(T)) \leq \fvs(H')$ so it remains to show that~$\fvs(H' - V(T)) \geq \fvs(H')$. Suppose that~$\fvs(H' - V(T)) < \fvs(H')$. Let~$X$ be a FVS in~$H'-V(T)$ with~$|X| < \fvs(H') = |C \cap V(H')| \leq z$. For any~$v \in C \cap V(H')$ such that~$H[V(T) \cup \{v\}]$ has a cycle we know from \cref{item:vflower} that~$H'-V(T)$ has~$z > |X|$ cycles that intersect only in~$v$, hence~$v \in X$. By \cref{item:pumpkin} we have that all but possibly one vertex in~$N_H(T)$ must be contained in~$X$, since if there are two vertices~$x,y \in N_H(T)\setminus X$ then~$H-V(T)-X$ has at least~$z+1 - |X| \geq 2$ internally vertex-disjoint paths between~$x$ and~$y$ forming a cycle and contradicting our choice of~$X$. Since there is at most one vertex~$w \in N_H(T)\setminus X$ and~$H[T \cup \{w\}]$ does not have a cycle, we have that~$H'-X$ is acyclic, a contradiction since~$|X| < \fvs(H')$.
 
 The desired $C$-certificate~$\hat{H}$ can be obtained from~$H$ by iteratively removing trees from~$H-C$ for which both conditions hold. We show that if no such tree exists, then~$H-C$ has at most~$\frac{|C|}{2}(z^2+2z-1)$ trees. Each tree~$T$ for which \cref{item:vflower} fails can be charged to a vertex~$v \in C$ that witnesses this, i.e.,~$H[T \cup \{v\}]$ has a cycle and there are at most~$z$ trees~$T'$ such that~$T' \cup \{v\}$ has a cycle. Clearly each vertex~$v \in C$ can be charged at most~$z$ times, hence there are at most~$z \cdot |C|$ trees violating \cref{item:vflower}.
 Similarly each tree~$T$ for which \cref{item:pumpkin} fails can be charged to a pair of vertices~$\{u,v\} \in \binom{N_H(T)}{2}$ for which at most~$z+1$ trees in~$H-T$ are adjacent to~$u$ and~$v$. Clearly each pair of vertices can be charged at most~$z+1$ times. Additionally each pair consists of vertices from the same component of~$H$. Let~$H_1,\ldots,H_\ell$ be the components in~$H$, then the number of such pairs is at most
 \begin{align*}
  \sum_{1 \leq i \leq \ell} \binom{|C \cap V(H_i)|}{2}
  &= \sum_{1 \leq i \leq \ell} \frac{1}{2}|C \cap V(H_i)|(|C \cap V(H_i)|-1)\\
  &\leq \sum_{1 \leq i \leq \ell} \frac{1}{2}|C \cap V(H_i)|(z-1)\\
  &= \frac{|C|}{2}(z-1).
 \end{align*}
 Thus~$H-C$ has at most~$z\cdot|C| + (z+1)\cdot\frac{|C|}{2}(z-1) = \frac{|C|}{2}(z^2+2z-1)$ trees violating \cref{item:pumpkin}.
 
 To conclude, there are at most~$z \cdot |C| + \frac{|C|}{2}(z^2+2z-1)$ trees of~$H-C$ that violate one or both conditions and any tree in~$H-C$ that satisfies both conditions can be excluded from~$\hat{H}$. We obtain a final bound on the number of trees in~$\hat{H}$ of~$z \cdot |C| + \frac{|C|}{2}(z^2+2z-1) = \frac{|C|}{2}(z^2+4z-1)$.
\end{proof}

We can now give an upper bound on the number of trees in~$G[F]$ required for a $z$-antler~$(C,F)$.

\begin{lemma}\label{col:few-fvcs}
 Let~$(C,F)$ be a $z$-antler in a multigraph~$G$ for some~$z \geq 0$. There exists an~$F' \subseteq F$ such that~$(C,F')$ is a $z$-antler in~$G$ and~$G[F']$ has at most~$\frac{|C|}{2}(z^2+4z-1)$ trees.
\end{lemma}
\begin{proof}
 Since~$(C,F)$ is a $z$-antler,~$G[C \cup F]$ contains a $C$-certificate~$H$ of order~$z$. By \cref{lem:small-cert} we know~$H$ contains a $C$-certificate~$\hat{H}$ of order~$z$ such that~$\hat{H}-C$ has at most~$\frac{|C|}{2}(z^2+4z-1)$ components. Since~$\hat{H}$ is a subgraph of~$H$, which is a subgraph of~$G[C \cup F]$, each component~$T_{\hat{H}}$ of~$\hat{H} - C$ (which is a tree) is contained in a component~$T_F$ of~$G[F]$ (which is also a tree). By the bound on the number of components of~$\hat{H}-C$, there are at most~$\frac{|C|}{2}(z^2+4z-1)$ trees of~$G[F]$ that contain a tree of~$\hat{H}-C$. Let~$F'$ denote the union of the vertex sets of these trees in~$G[F]$. Since each tree of~$G[F']$ is also a tree of~$G[F]$, the resulting pair~$(C,F')$ is a FVC. As the $C$-certificate~$\hat{H}$ of order~$z$ is a subgraph of~$G[C \cup F']$, we conclude~$(C,F')$ is the desired $z$-antler.
\end{proof}

\section{Finding and reducing feedback vertex cuts} \label{sec:find-fvc}
As described in \cref{antler:sec:intro}, our algorithm aims to identify vertices in antlers using color coding. To allow a relatively small family of colorings to identify an entire antler structure~$(C,F)$ with~$|C| \leq k$, we need to bound~$|F|$ in terms of~$k$ as well. We therefore use several graph reduction steps. In this section, we show that if there is a width-$k$ antler whose forest~$F$ is significantly larger than~$k$, then we can identify a reducible structure in the graph. To identify a reducible structure we will also use color coding. In \cref{sec:fvc-kernel} we show how to reduce such a structure while preserving antlers and optimal feedback vertex sets. 

We define several types of special feedback vertex cuts below. These concepts will be used to reason about feedback vertex cuts that can be found and manipulated efficiently.

\begin{definition} \label{def:reducible:fvc}
Define the function~$f_r \colon \mathbb{N} \to \mathbb{N}$ as~$f_r(x) = 2x^3+4x^2$. We say~$(C,F)$ is a \emph{single-tree} FVC if~$G[F]$ is connected. A FVC~$(C,F)$ is \emph{reducible} if the following conditions hold:~$|F| > f_r(|C|)$ and~$|N_G(F)| \leq |C|+1$.
\end{definition}

Note that the condition~$|N_G(F)| \leq |C|+1$ trivially holds for any single-tree FVC. We will show that, given a reducible FVC~$(C,F)$, we can efficiently reduce to a smaller graph in which to continue the algorithm. However, it is not trivial to find a large FVC in a graph. The next definition will be used to give guarantees about efficiently finding a large FVC.


\begin{definition} \label{def:simple}
 A FVC~$(C,F)$ is \emph{simple} if~$|F| \leq 2f_r(|C|)$ and one of the following holds:
 \begin{enumerate}[label=(\alph*)]
     \item $G[F]$ is connected, or \label[condition]{condition:connected}
     \item all trees in~$G[F]$ have a common neighbor~$v$ in~$G$ and there exists a single-tree FVC~$(C,F_1)$ with~$v \in F_1 \setminus F$ and~$F \subseteq F_1$.\label[condition]{condition:common:neighbor}
\end{enumerate}
\end{definition}

The intuition behind this definition is as follows. To find a FVC~$(C,F)$ in a multigraph~$G$ using color coding, we enumerate a number of colorings of the vertex set of~$G$ in the hope that at least one of them ``highlights''~$(C,F)$ in a certain technical sense. Observe that if~$G[F]$ consists of a single tree~$T$ and we have a coloring in which~$V(T)$ receives one color~$\cF$, and all vertices of~$N_G(T)$ receive a different color~$\cC$, then given this coloring it is easy to find~$T$ since it appears as a connected component of the subgraph induced by the vertices with color~$\cF$. Given~$T$, it is easy to find a FVC~$(C,F)$ satisfying~$F=V(T)$ by simply setting~$C = N_G(T)$. There may even be a FVC with~$F=V(T)$ for which the set~$C$ is strictly smaller. This happens precisely when there is a vertex~$v \in N_G(T)$ with~$e(T,v)=1$, since then we may set~$C = N_G(T) \setminus \{v\}$ and still obtain a valid FVC. Hence a single-tree FVC is particularly easy to find due to its forest appearing as a single connected component of the graph induced by the suitable color.

When a FVC~$(C,F)$ has multiple trees in its forest~$G[F]$, it becomes harder to detect; even when we have a coloring in which all of~$F$ receives one color and all of~$N_G(F)$ receives another. The reason for this is that it is a priori unclear which trees should be combined to form the forest; this is precisely the difficulty underlying the hardness proof of \cref{thm:1antler:w1hard}. But when a FVC is simple via \cref{condition:common:neighbor}, the common neighbor~$u$ acts as an anchor point of the trees which facilitates the process of assembling the forest~$F$ out of a suitable collection of trees. The polynomial-time dynamic program for the 0-1 knapsack problem will play a role here.

Based on these ideas, we will present an algorithm that can identify a reducible FVC when it is also simple. First we show that such a FVC always exists when the graph contains a single-tree reducible FVC.

\begin{lemma}\label{lem:small_fvc}
 If a multigraph~$G$ contains a reducible single-tree FVC~$(C,F)$, then there exists a simple reducible FVC~$(C,F')$ with~$F' \subseteq F$.
\end{lemma}
\begin{proof}
 We use induction on~$|F|$. If~$|F| \leq 2f_r(|C|)$ then~$(C,F)$ is simple by \cref{condition:connected}. Assume~$|F| > 2f_r(|C|)$. Since~$(C,F)$ is a FVC and~$G[F]$ is connected there is at most one vertex~$v \in F$ that has a neighbor in~$V(G)\setminus(C \cup F)$. If no such vertex exists, take~$v \in F$ to be any other vertex. Observe that~$(C,F\setminus\{v\})$ is a FVC:~$v$ was the only vertex of the tree~$G[F]$ that had a neighbor outside~$C$ (if any), so that for each tree of~$G[F] - v$, vertex~$v$ is the only neighbor outside~$C$. Consider the following cases:
 \begin{itemize}
  \item All trees in~$G[F] - v$ contain at most~$f_r(|C|)$ vertices. Let~$F'$ be the vertices of an inclusion minimal set of trees of~$G[F]-v$ such that~$|F'| > f_r(|C|)$. Clearly~$|F'| \leq 2f_r(|C|)$ since otherwise the set is not inclusion minimal. Each tree in~$G[F']$ contains a neighbor of~$v$ and~$F' \subseteq F$, hence~$(C,F')$ is simple by \cref{condition:common:neighbor}, and~$(C,F')$ is reducible since~$|F'| > f_r(|C|)$ while~$N_G(F')\subseteq C \cup \{v\}$.
  \item There is a tree~$T$ in~$G[F] - v$ that contains more than~$f_r(|C|)$ vertices. Now~$(C,V(T))$ is a single-tree reducible FVC with~$|V(T)| < |F|$, so the result follows by induction.
  \qedhere
 \end{itemize}
\end{proof}

We proceed to show how a simple reducible FVC can be found using color coding. A vertex coloring of~$G$ is a function~$\chi \colon V(G) \to \{\cC,\cF\}$. We say a simple FVC~$(C,F)$ is \emph{properly colored} by a coloring~$\chi$ if~$F \subseteq \inv{\chi}(\cF)$ and~$C \cup N_G(F) \subseteq \inv{\chi}(\cC)$.

\begin{lemma} \label{lem:colorcoded-fvc}
Given a multigraph~$G$ and coloring~$\chi$ of~$G$ that properly colors some simple reducible FVC~$(C,F)$, a reducible FVC~$(C',F')$ can be found in~$\Oh(n^3)$ time.
\end{lemma}
\begin{proof}
If~$(C,F)$ is simple by \cref{condition:connected}, i.e.,~$G[F]$ is connected, it is easily verified that we can find a reducible FVC in~$\Oh(n^3)$ time as follows: Consider the set~$\mathcal{T}$ of all trees in~$G[\inv{\chi}(\cF)]$. For each tree~$T \in \mathcal{T}$, if there is a vertex~$u \in N_G(T)$ such that~$e(\{u\},T) = 1$ take~$C' := N_G(T) \setminus \{u\}$, otherwise take~$C' := N_G(T)$. If~$|V(T)| > f_r(|C'|)$ return~$(C',V(T))$.

In the remainder of the proof we assume that~$(C,F)$ is simple by \cref{condition:common:neighbor}.

\subparagraph*{Algorithm}
For each vertex~$u \in \inv{\chi}(\cC)$ consider the set~$\mathcal{T}$ of all trees~$T$ in~$G[\inv{\chi}(\cF)]$ such that~$e(\{u\},T) = 1$. Let~$C' \subseteq \inv{\chi}(\cC) \setminus \{u\}$ be the set of vertices (excluding~$u$) with a neighbor in at least two trees in~$\mathcal{T}$ and let~$\mathcal{T}_1$ be the set of trees~$T \in \mathcal{T}$ for which~$N_G(T) \subseteq C' \cup \{u\}$. Now consider the set of trees~$\mathcal{T}_2 = \mathcal{T}\setminus\mathcal{T}_1$ as a set of objects for a 0-1 knapsack problem where we define for each~$T \in \mathcal{T}_2$ its weight as~$|N_G(T)\setminus(C' \cup \{u\})|$ and its value as~$|V(T)|$. Using the dynamic programming algorithm~\cite{Toth80} for the 0-1 knapsack problem we compute for all~$0 \leq b \leq |N_G(V(\mathcal{T}_2)) \setminus (C' \cup \{u\})|$ a set of trees~$\mathcal{T}_2^b \subseteq \mathcal{T}_2$ with a combined weight~$\sum_{T \in \mathcal{T}_2^b} |N_G(T) \setminus (C' \cup \{u\})| \leq b$ such that the combined value~$\sum_{T \in \mathcal{T}_2^b} |V(T)|$ is maximized. If for any such~$b$ we have~$|V(\mathcal{T}_1)|+ |V(\mathcal{T}_2^b)| > f_r(|C'| + b)$ then take~$\hat{C} := C' \cup N_G(V(\mathcal{T}_2^b)) \setminus \{u\}$ and~$\hat{F} := V(\mathcal{T}_1) \cup V(\mathcal{T}_2^b)$ and return~$(\hat{C},\hat{F})$.

\subparagraph*{Correctness}
To show that~$(\hat{C}, \hat{F})$ is a FVC, first note that~$G[\hat{F}]$ is a forest. For each tree~$T$ in this forest we have~$e(T,\{u\})=1$ and~$N_G(T) \subseteq C' \cup \{u\} \cup N_G(V(\mathcal{T}_2^b)) = \hat{C} \cup \{u\}$. It follows that~$e(T,G-(\hat{C} \cup \hat{F}))=e(T,\{u\}) = 1$. To show that~$(\hat{C},\hat{F})$ indeed reducible observe that~$\sum_{T \in \mathcal{T}_2^b} |N_G(T) \setminus (C' \cup \{u\})| = |\bigcup_{T \in \mathcal{T}_2^b} N_G(T) \setminus (C' \cup \{u\})|$ since if two trees~$T_1,T_2 \in \mathcal{T}_2^b$ have a common neighbor that is not~$u$, it must be in~$C'$ by definition, hence the neighborhoods only intersect on~$C' \cup \{u\}$. We can now deduce~$|\hat{F}| = |V(\mathcal{T}_1)| + |V(\mathcal{T}_2^b)| > f_r(|C'| + b) \geq f_r(|C'| + \sum_{T \in \mathcal{T}_2^b} |N_G(T) \setminus (C' \cup \{u\})|) = f_r(|C' \cup N_G(V(\mathcal{T}_2^b)) \setminus \{u\}|) = f_r(|\hat{C}|)$. The second criterion of \cref{def:reducible:fvc} is satisfied since~$N_G(\hat{F}) \subseteq \hat{C} \cup \{u\}$.

It remains to show that if~$\chi$ properly colors a simple reducible FVC~$(C,F)$ then for some~$u \in \inv{\chi}(\cC)$ there exists a~$b$ such that~$|V(\mathcal{T}_1)| + |V(\mathcal{T}_2^b)| \geq f_r(|C'| + b)$. Recall that we assumed~$(C,F)$ is simple by \cref{condition:common:neighbor}, i.e., all trees in~$G[F]$ have a common neighbor~$v$ and there exists a single-tree FVC~$(C,F_2)$ with~$v \in F_2 \setminus F$ and~$F \subseteq F_2$. Since~$(C,F)$ is properly colored we know~$v \in \inv{\chi}(\cC)$, so in some iteration we will have~$u = v$. Consider the sets~$\mathcal{T}$,~$\mathcal{T}_1$,~$\mathcal{T}_2$, and~$C'$ as defined in this iteration. We first show~$C' \subseteq C$. If~$w \in C'$ then~$w$ has a neighbor in two trees in~$\mathcal{T}$. This means there are two internally vertex disjoint paths between~$v$ and~$w$, forming a cycle. Since~$v \in F_2$ we have by \cref{obs:fvc-basics} for the FVC~$(C,F_2)$ that this cycle must contain a vertex in~$C$ which is therefore different from~$v$. Recall that~$(C,F)$ is properly colored, hence all vertices in~$C$ have color~\cC. Note that the internal vertices of these paths all have color~\cF because they are vertices from trees in~$G[\inv{\chi}(\cF)]$. Hence~$w \in C$ and therefore~$C' \subseteq C$.
To complete the proof we show the following.

\begin{claim}
 There exists a value~$b$ such that~$|V(\mathcal{T}_1)| + |V(\mathcal{T}_2^b)| \geq f_r(|C'| + b)$.
\end{claim}
\begin{claimproof}
 Recall that we assumed existence of a properly colored FVC~$(C,F)$ that is reducible and simple by \cref{condition:common:neighbor} witnessed by the FVC~$(C,F_2)$. Consider the set~$\mathcal{T}'$ of trees in~$G[F]$. Note that any tree~$T'$ in~$\mathcal{T}'$ is a tree in~$G[\inv{\chi}(\cF)]$ since~$(C,F)$ is properly colored and note also that~$T'$ contains a neighbor of~$v$. If~$e(T',\{v\}) > 1$ then~$G[F_2]$ contains a cycle, contradicting that~$(C,F_2)$ is a FVC in~$G$, hence~$e(T',\{v\})=1$. It follows that~$T' \in \mathcal{T}$, meaning~$\mathcal{T'} \subseteq \mathcal{T}$. Take~$\mathcal{T}'_2 = \mathcal{T}' \setminus \mathcal{T}_1 = \mathcal{T}' \cap \mathcal{T}_2$ and~$b = \sum_{T \in \mathcal{T}'_2} |N_G(V(T)) \setminus (C' \cup \{v\}|$. Clearly~$\mathcal{T}'_2$ is a candidate solution for the 0-1 knapsack problem with capacity~$b$, hence~$|V(\mathcal{T}_2^b)| \geq |V(\mathcal{T}'_2)|$. We deduce
\allowdisplaybreaks
\begin{align*}
 |V(\mathcal{T}_1)| + |V(\mathcal{T}_2^b)|
 &\geq |V(\mathcal{T}_1)| + |V(\mathcal{T}'_2)|
 \geq |V(\mathcal{T}')|
 = |F|\\
 &\hspace{-5em}> f_r(|C|) &&\hspace{-5em}\text{since~$(C,F)$ is reducible}\\
 &\hspace{-5em}= f_r(|C' \cup C|) &&\hspace{-5em}\text{since~$C' \subseteq C$}\\
 &\hspace{-5em}= f_r(|C' \cup (N_G(F) \setminus \{v\}) \cup C|) &&\hspace{-5em}\text{since~$N_G(\mathcal{T}'_2) \setminus \{v\} \subseteq C$}\\
 &\hspace{-5em}\geq f_r(|C' \cup (N_G(\mathcal{T}'_2) \setminus \{v\})|) &&\hspace{-5em}\text{since~$f_r$ is non-decreasing}\\
 &\hspace{-5em}= f_r(|C'| + |N_G(\mathcal{T}'_2) \setminus (C' \cup \{v\})|) &&\hspace{-5em}\text{since~$|A \cup B| = |A| + |B \setminus A|$}\\
 &\hspace{-5em}> f_r(|C'| + b). 
 \tag*{\qedhere}
\end{align*}
\end{claimproof}

\subparagraph*{Running time}
For each~$u \in \inv{\chi}(\cC)$ we perform a number of $\Oh(n+m)$-time operations and run the dynamic programming algorithm for a problem with~$\Oh(n)$ objects and a capacity of~$\Oh(n)$ yielding a run time of~$\Oh(n^2)$ for each~$u$ or~$\Oh(n^3)$ for the algorithm as a whole.
\end{proof}

Using the previous lemmas the problem of finding a reducible single-tree FVC reduces to finding a coloring that properly colors a simple reducible FVC. We generate a set of colorings that is guaranteed to contain at least one such coloring. To generate this set we use the concept of a universal set.
For some set~$D$ of size~$n$ and integer~$s$ with~$n \geq s$, an~\emph{$(n,s)$-universal set} for~$D$ is a family~$\mathcal{U}$ of subsets of~$D$ such that for all~$S \subseteq D$ of size at most~$s$ we have~$\{S \cap U \mid U \in \mathcal{U}\} = 2^S$.

\begin{theorem}[{\cite[Theorem~6]{NaorSS95}, cf.~\cite[Theorem~5.20]{CyganFKLMPPS15}}] \label{thm:universal-set}
 For any set~$D$ and integers~$n$ and~$s$ with~$|D| = n \geq s$, an $(n,s)$-universal set~$\mathcal{U}$ for~$D$ with~$|\mathcal{U}| = 2^{\Oh(s)}\log n$ can be created in~$2^{\Oh(s)} n \log n$ time.
\end{theorem}

Whether a simple FVC of width~$k$ is properly colored is determined by at most~$2f_r(k) + k + 1 = \Oh(k^3)$ relevant vertices. By creating an $(n,\Oh(k^3))$-universal set for~$V(G)$ using \cref{thm:universal-set}, we can obtain in~$2^{\Oh(k^3)} \cdot n \log n$ time a set of~$2^{\Oh(k^3)} \cdot \log n$ colorings that contains a coloring for each possible assignment of colors for these relevant vertices. By applying \cref{lem:colorcoded-fvc} for each coloring we obtain the following lemma.

\begin{lemma}\label{lem:find-fvc}
 There exists an algorithm that, given a multigraph~$G$ and an integer~$k$, outputs a (possibly empty) FVC~$(C,F)$ in~$G$. If~$G$ contains a reducible single-tree FVC of width at most~$k$, then~$(C,F)$ is reducible. The algorithm runs in time~$2^{\Oh(k^3)} \cdot n^3 \log n$.
 \end{lemma}
 \begin{proof}
 Take~$s = 2f_r(k) + k + 1$. By \cref{thm:universal-set} an $(n,s)$-universal set~$\mathcal{U}$ for~$V(G)$ of size~$2^{\Oh(s)} \log n$ can be created in~$2^{\Oh(s)} n \log n$ time. For each~$Q \in \mathcal{U}$, let~$\chi_Q$ be the coloring of~$G$ with~$\inv{\chi_Q}(\cC) = Q$. Run the algorithm from \cref{lem:colorcoded-fvc} on~$\chi_Q$ for every~$Q \in \mathcal{U}$ and return the first reducible FVC. If no reducible FVC was found return~$(\emptyset,\emptyset)$.
 We obtain an overall run time of~$2^{\Oh(s)} \cdot n^3 \log n = 2^{\Oh(k^3)} \cdot n^3 \log n$.
 
 To prove correctness assume~$G$ contains a reducible single-tree FVC~$(C_1,F_1)$ with~$|C_1| \leq k$. Then by \cref{lem:small_fvc} we know~$G$ contains a simple reducible FVC~$(C_1,F_2)$. Coloring~$\chi$ properly colors~$(C_1,F_2)$ if all vertices in~$F_2 \cup C_1 \cup N_G(F_2)$ are assigned the correct color. Hence at most~$|F_2| + |C_1 \cup N_G(F_2)| \leq 2f_r(k) + k + 1 = s$ vertices need to have the correct color; here we use that \cref{def:simple} ensures~$|C_1 \cup N_G(F_2)| \leq k+1$. By construction of~$\mathcal{U}$, there is a~$Q \in \mathcal{U}$ such that~$\chi_Q$ assigns the correct colors to these vertices. Hence~$\chi_Q$ properly colors~$(C_1,F_2)$ and by \cref{lem:colorcoded-fvc} a reducible FVC is returned.
\end{proof}

\subsection{Reducing feedback vertex cuts} \label{sec:fvc-kernel}
We introduce reduction rules inspired by existing kernelization algorithms~\cite{BodlaenderD10, Thomasse10} and apply them on the subgraph~$G[C \cup F]$ for a FVC~$(C,F)$ in~$G$. We give 5 reduction rules and show at least one is applicable if~$|F| > f_r(|C|)$. The rules reduce the number of vertices~$v \in F$ with~$\deg_G(v) < 3$ or reduce~$e(C,F)$. The following lemma shows that this is sufficient to reduce the size of~$F$.

\begin{lemma}\label{lem:degreesum}
 Let~$G$ be a multigraph with minimum degree at least~$3$ and let~$(C,F)$ be a FVC in~$G$. We have~$|F| \leq e(C,F)$.
\end{lemma}
\begin{proof}
 We first show that the claim holds if~$G[F]$ is a tree. For all~$i \geq 0$ let~$V_i := \{v \in F \mid \degree_{G[F]}(v) = i\}$. Note that since~$G[F]$ is connected,~$V_0 \neq \emptyset$ if and only if~$|F|=1$ (and then the claim is trivially true). We therefore assume~$V_0 = \emptyset$. We first show~$|V_{\geq 3}| < |V_1|$, using the fact that in any tree the number of vertices is one larger than the number of edges.
 \begin{align*}
  2|E(G[F])| &= \sum_{v \in F} \degree_{G[F]}(v) \geq |V_1| + 2|V_2| + 3|V_{\geq 3}| & \text{$v\in V_i$ has degree~$i$} \\
  2|E(G[F])| &= 2(|V(G[F])| - 1) = 2|V_1| + 2|V_2| + 2|V_{\geq 3}| - 2 & \text{since~$G[F]$ is a tree}
 \end{align*}
 We obtain~$|V_1| + 2|V_2| + 3|V_{\geq 3}| \leq 2|V_1| + 2|V_2| + 2|V_{\geq 3}| - 2$ hence~$|V_{\geq 3}| < |V_1|$.
 By assumption, all vertices in~$F$ have degree at least~$3$ in~$G$, so each vertex of~$V_1$ is incident on at least~$2$ edges outside~$G[F]$ and each vertex of~$V_2$ is incident to at least~$1$. We therefore find~$e(V(G) \setminus F,F) \geq 2|V_1| + |V_2| > |V_1| + |V_2| + |V_{\geq 3}| = |F|$. By definition of FVC there is at most one vertex in~$F$ that has an edge to~$V(G) \setminus (C \cup F)$, all other edges must be between~$C$ and~$F$. We obtain~$1+e(C,F) > |F|$ for the case that~$G[F]$ is a tree.
 
 If~$G[F]$ is a forest, then let~$F_1, \ldots, F_\ell$ be the vertex sets of the trees in~$G[F]$. Since~$(C,F_i)$ is a FVC in~$G$ for all~$1 \leq i \leq \ell$, we know~$e(C,F_i) \geq |F_i|$ for all~$1 \leq i \leq \ell$, and since~$F_1, \ldots, F_\ell$ is a partition of~$F$ we conclude~$e(C,F) = \sum_{1 \leq i \leq \ell} e(C,F_i) \geq \sum_{1 \leq i \leq \ell} |F_i| = |F|$.
\end{proof}

Next, we give the reduction rules. These rules apply to a multigraph~$G$ and yield a new multigraph~$G'$ and vertex set~$S \subseteq V(G)\setminus V(G')$.

\begin{definition} \label{def:safe}
A reduction rule with output~$G'$ and~$S$ is \emph{FVS-safe} if for any minimum feedback vertex set~$S'$ of~$G'$, the set~$S \cup S'$ is a minimum feedback vertex set of~$G$.

A reduction rule is \emph{antler-safe} if for all~$z\geq0$ and any $z$-antler~$(C,F)$ in~$G$, there exists a $z$-antler~$(C',F')$ in~$G'$ with~$C' \cup F' = (C \cup F) \cap V(G')$ and~$|C'| = |C| - |(C \cup F) \cap S|$.
\end{definition}

The first type of safety ensures that finding vertices that belong to an optimal FVS of~$G'$ leads to finding vertices that belong to an optimal FVS of~$G$. The second type of safety ensures that if, in the original graph~$G$, there is a simple certificate that some~$k$ vertices belong to an optimal solution (in the form of a $z$-antler of width~$k$), then in~$G'$ there is still a simple certificate for the membership of~$k-|S|$ vertices in an optimal solution. (The reduction step already identified~$|S|$ of them.)

\begin{reduction} \label{op:edge-multiplicity}
 If~$u,v \in V(G)$ are connected by more than two edges, remove all but two of these edges to obtain~$G'$ and take~$S := \emptyset$.
\end{reduction}
\begin{reduction} \label{op:contract-edge}
 If~$v \in V(G)$ has degree exactly~$2$ and no self-loop, obtain~$G'$ by removing~$v$ from~$G$ and adding an edge~$e$ with~$\iota(e) = N_G(v)$. Take~$S := \emptyset$.
\end{reduction}

\Cref{op:edge-multiplicity,op:contract-edge} are well established and FVS-safe. Additionally \cref{op:edge-multiplicity} can easily be seen to be antler-safe. To see that \cref{op:contract-edge} is antler-safe, consider a $z$-antler~$(C,F)$ in~$G$ for some~$z\geq0$. If~$v \not\in C$ it is easily verified that~$(C,F\setminus\{v\})$ is a $z$-antler in~$G'$. If~$v \in C$ pick a vertex~$u \in N_G(v) \cap F$ and observe that~$(\{u\} \cup C \setminus \{v\}, F \setminus \{u\})$ is a $z$-antler in~$G'$.

The following rule will be useful in several scenarios. We will not be able to apply it exhaustively.

\begin{reduction} \label{op:remove-antler}
 If~$(C,F)$ is an antler in~$G$, then set~$G' := G - (C \cup F)$ and~$S := C$.
\end{reduction}
\begin{lemma}
 \Cref{op:remove-antler} is FVS-safe and antler-safe.
\end{lemma}
\begin{proof}
 To show \cref{op:remove-antler} is FVS-safe, let~$Z$ be a minimum FVS of~$G'$. Then~$Z \cup S$ is a FVS of~$G$, since~$Z$ breaks all cycles of~$G'$ and~$C$ breaks any cycle intersecting~$C \cup F$. This FVS is minimum since~$|C| = \fvs(G[C \cup F])$ by definition of antler. 
 
 To show \cref{op:remove-antler} is antler-safe, let~$z\geq0$ and let~$(\hat{C},\hat{F})$ be an arbitrary $z$-antler in~$G$. By \cref{lem:antler_diff}, the pair~$(\hat{C}\setminus(C \cup F), \hat{F}\setminus(C \cup F))$ is a $z$-antler in~$G' = G-(C \cup F)$. We deduce:
 \begin{align*}
  |\hat{C} \setminus (C \cup F)|
  &= |\hat{C}| - |\hat{C} \cap C| - |\hat{C} \cap F|
    &&\text{since $C \cap F = \emptyset$}\\
  &= |\hat{C}| - |\hat{C} \cap C| - |C \cap \hat{F}|
    &&\text{by \cref{prop:intersection_antlers}}\\
  &= |\hat{C}| - |(\hat{C} \cap C) \cup (C \cap \hat{F})|
    &&\text{since $\hat{C} \cap \hat{F} = \emptyset$}\\
  &= |\hat{C}| - |(\hat{C} \cup \hat{F}) \cap C|\\
  &= |\hat{C}| - |(\hat{C} \cup \hat{F}) \cap S|.
  \tag*{\qedhere}
 \end{align*}
\end{proof}

The following reduction rule uses a graph structure called flower. For a vertex~$v \in V(G)$ and an integer~$k$, a \emph{$v$-flower} of \emph{order}~$k$ in a multigraph~$G$ is a collection of~$k$ cycles in~$G$ whose vertex sets pairwise intersect exactly in~$v$.

\begin{reduction} \label{op:v-flower}
 If~$(C,F)$ is a FVC in~$G$ and for some~$v \in C$  the subgraph~$G[F \cup \{v\}]$ contains a $v$-flower of order~$|C|+1$, then set~$G' := G - v$ and~$S := \{v\}$.
\end{reduction}
\begin{lemma}
 \Cref{op:v-flower} is FVS-safe and antler-safe.
\end{lemma}
\begin{proof}
 We first show that any minimum FVS in~$G$ contains~$v$. Let~$X$ be a minimum FVS in~$G$. If~$v \not\in X$ then~$|F \cap X| > |C|$ since~$G[F \cup \{v\}]$ contains a $v$-flower of order~$|C|+1$. Take~$X' := C \cup (X \setminus F)$, clearly~$|X'| < |X|$ so~$G-X'$ must contain a cycle since~$X$ was minimum. This cycle must contain a vertex from~$X \setminus X' \subseteq F$, so by \cref{obs:fvc-basics} this cycle must contain a vertex from~$C$, but~$C \subseteq X'$. Contradiction.
 
 To show \cref{op:v-flower} is FVS-safe, suppose~$Z$ is a minimum FVS of~$G' = G-v$. Clearly~$Z \cup \{v\}$ is a FVS in~$G$. To show that~$Z \cup \{v\}$ is minimum, suppose~$Z'$ is a smaller FVS in~$G$. We know~$v \in Z$ so~$Z'\setminus\{v\}$ is a FVS in~$G-v$, but~$|Z'\setminus\{v\}| < |Z|$ contradicting optimality of~$Z$.
 
 To show \cref{op:v-flower} is antler-safe, suppose~$(\hat{C},\hat{F})$ is a $z$-antler in~$G$ for some~$z\geq0$. We show~$(\hat{C}\setminus\{v\}, \hat{F})$ is a $z$-antler in~$G'$. If~$v \in \hat{C}$ then this follows directly from \cref{obs:subantler}, so suppose~$v \not\in \hat{C}$. Note that~$v \in \hat{F}$ would contradict that any minimum FVS in~$G$ contains~$v$ by \cref{obs:remove-antler}. So~$G[\hat{C} \cup \hat{F}] = G'[\hat{C} \cup \hat{F}]$ and~$(\hat{C}\setminus\{v\}, \hat{F}) = (\hat{C},\hat{F})$ is a FVC in~$G' = G-v$ by \cref{obs:subfvc}, hence~$(\hat{C}\setminus\{v\}, \hat{F})$ is a $z$-antler in~$G'$.
 %
\end{proof}

The final reduction rule is the most technical. The underlying ideas are related to those in \cref{lem:small-cert}, but we target a stronger outcome. Rather than just bounding the number of trees in~$G[F]$, the goal is to reduce the number of edges that a given vertex~$v \in C$ has into~$F$. In the setting that~$v$ has many neighbors in a single tree of~$G[F]$, this requires additional work compared to \cref{lem:small-cert}. 

In the following rule and its analysis, we refer to an acyclic connected component~$T$ of a graph~$H$ as a \emph{tree~$T$ in~$H$}; this should not be confused with the concept of an arbitrary subgraph of~$H$ that forms a tree.

\begin{reduction} \label{op:remove-edge}
 Let~$(C,F)$ be a FVC in~$G$,~$v \in C$, and~$X \subseteq F$ such that~$G[F \cup \{v\}] - X$ is acyclic. If~$T$ is a tree in~$G[F]-X$ containing a vertex~$w \in N_G(v)$ such that:
 \begin{itemize}
     \item for each~$u \in N_G(T)\setminus\{v\}$ there are more than~$|C|$ other trees~$T' \neq T$ in~$G[F]-X$ for which~$\{u,v\} \subseteq N_G(T')$,
 \end{itemize}
 then take~$S := \emptyset$ and obtain~$G'$ by removing the unique edge between~$v$ and~$w$, and adding double-edges between~$v$ and~$u$ for all~$u \in N_G(V(T))\setminus\{v\}$.
\end{reduction}
\begin{lemma}
 \Cref{op:remove-edge} is FVS-safe and antler-safe.
\end{lemma}
\begin{proof}
 Assume the stated conditions hold. We first prove the following claim: 
 
  \begin{claim}\label{claim:op4:disjointpaths}
   Let~$(\hat{C},\hat{F})$ be a $z$-antler in~$G$ for some~$z\geq0$. For all~$u \in N_G(T)\setminus\{v\}$ we have:
   \begin{itemize}
       \item $v \in \hat{F}$ implies~$u \in \hat{C}$, and
       \item $u \in \hat{F}$ implies~$v \in \hat{C}$.
   \end{itemize}
  \end{claim}
  \begin{claimproof}
  Each tree~$T'$ of~$G[F]-X$ with~$\{u,v\} \subseteq N_G(T')$ supplies a path between~$u$ and~$v$, hence there are more than~$|C|+1$ internally vertex-disjoint paths between~$u$ and~$v$ in~$G$. We prove that~$v \in \hat{F}$ implies~$u \in \hat{C}$; the proof of the second implication is symmetric. 
  
  Assume~$v \in \hat{F}$ and suppose for contradiction that~$u \not\in \hat{C}$. All except possibly one of the disjoint paths between~$u$ and~$v$ must contain a vertex in~$\hat{C}$ by \cref{obs:fvc-basics}, since any two disjoint paths form a cycle containing a vertex from~$\hat{F}$. Let~$Y \subseteq \hat{C}$ be the set of vertices in~$\hat{C}$ that are in a tree~$T'$ of~$G[F]-X$ with~$\{u,v\} \subseteq N_G(T')$, so~$|Y| > |C|$. Then~$|(C \cup \hat{C}) \setminus Y| < |\hat{C}|$. The graph~$G[\hat{C} \cup \hat{F}] - ((C \cup \hat{C}) \setminus Y)$ is acyclic, since any cycle in~$G$ containing a vertex from~$Y \subseteq F$ also contains a vertex from~$C$ by \cref{obs:fvc-basics}. This contradicts that~$\hat{C}$ is a (minimum) FVS in~$G[\hat{C} \cup \hat{F}]$, which follows from the assumption that~$(\hat{C},\hat{F})$ is an antler in~$G$.
  \end{claimproof}

 \subparagraph*{Antler-safe} We now prove that \cref{op:remove-edge} is antler-safe. Suppose~$(\hat{C},\hat{F})$ is a $z$-antler in~$G$ for some~$z\geq0$; we show that~$(\hat{C},\hat{F})$ is also a $z$-antler in~$G'$. We distinguish two cases, depending on the status of~$v$. Suppose first that~$v \not\in \hat{C} \cup \hat{F}$, then~$G[\hat{C} \cup \hat{F}] = G'[\hat{C} \cup \hat{F}]$ as~$G$ and~$G'$ only differ on edges incident to~$v$. It remains to show that for each tree~$T'$ in~$G'[\hat{F}]$ we have~$e_{G'}(T',G' - (\hat{C} \cup \hat{F})) \leq 1$. Suppose~$T'$ is a tree in~$G'[\hat{F}]$ with~$e_{G'}(T',G' - (\hat{C} \cup \hat{F})) > 1$. Since~$e_G(T',G - (\hat{C} \cup \hat{F})) \leq 1$ we know that at least one of the edges added between~$v$ and some~$u \in N_G(T)$ has an endpoint in~$V(T') \subseteq \hat{F}$. Since~$v \not\in \hat{F}$ we have~$u \in \hat{F}$, so~$v \in \hat{C}$ by \cref{claim:op4:disjointpaths} contradicting our assumption~$v \not\in \hat{C} \cup \hat{F}$.
 
 In the remainder, we may assume~$v \in \hat{C} \cup \hat{F}$. We first show that~$(\hat{C}, \hat{F})$ is a FVC in~$G'$. If~$v \in \hat{C}$ this is clearly the case, so suppose~$v \in \hat{F}$. From \cref{claim:op4:disjointpaths} it follows that~$N_G(T)\setminus\{v\} \subseteq \hat{C}$, so all edges added in~$G'$ are incident to vertices in~$\hat{C}$ hence~$(\hat{C},\hat{F})$ is still a FVC in~$G'$.
 
 We now show that~$G'[\hat{C} \cup \hat{F}]$ contains a $\hat{C}$-certificate of order~$z$. We know~$G[\hat{C} \cup \hat{F}]$ contains a $\hat{C}$-certificate of order~$z$. Let~$H$ be an arbitrary component of this certificate. We define a corresponding component~$H'$ for a~$\hat{C}$-certificate in~$G'[\hat{C} \cup \hat{F}]$, as follows.
 \begin{itemize}
     \item If~$v \notin V(H)$ or~$w \notin V(H)$, then define~$H' = H$ and observe that~$H'$ is a subgraph of~$G'$ since the only edge that was removed from~$G$ when building~$G'$ connects~$v$ to~$w$.
     \item If~$\{v,w\} \subseteq V(H)$, then let~$H'$ be the graph obtained from~$H$ by removing the edge between~$v$ and~$w$, while inserting a double-edge between~$v$ and all vertices~$u \in V(H) \cap (N_G(V(T)) \setminus \{v\})$. Note that~$H'$ is a subgraph of~$G'$.
 \end{itemize}
 Observe that we have~$V(H) = V(H')$ in both cases. To argue that the union of the components~$H'$ is the desired~$\hat{C}$-certificate of order~$z$, it suffices to consider each component individually and to prove that~$Y := V(H) \cap \hat{C} = V(H') \cap \hat{C}$ is a minimum FVS of~$H'$. Note that~$Y$ is a minimum FVS in~$H$ since~$H$ is a component of a~$\hat{C}$-certificate. 
 
 If~$v \not\in V(H)$ or~$w \notin V(H)$, then~$Y$ is trivially a minimum FVS of~$H'$ since~$H' = H$. So suppose~$\{v,w\} \in V(H)$. First we argue that~$H'-Y$ is acyclic. This is easily seen to be true when~$v \in Y$ since~$H$ and~$H'$ only differ in edges incident to~$v$, so suppose~$v \not\in Y$. Then~$v \not\in \hat{C}$ hence~$v \in \hat{F}$ and by \cref{claim:op4:disjointpaths} we have~$N_G(T)\setminus\{v\} \subseteq \hat{C}$. It follows that~$V(H) \cap (N_G(T)\setminus\{v\}) \subseteq Y$ so clearly~$H'-Y$ is acyclic since all edges in~$H'$ that are not in~$H$ are incident to a vertex in~$Y$, while~$H-Y$ is acyclic.
 
 To show~$Y$ is a \emph{minimum} FVS of~$H'$, suppose~$H'$ has a FVS~$Y'$ with~$|Y'| < |Y|$. Since~$H-v$ is a subgraph of~$H'$ we know~$H-(Y' \cup \{v\})$ is acyclic, but since~$|Y'| < |Y| = \fvs(H)$ we also know~$H-Y'$ contains a cycle. This cycle must contain the edge~$\{v,w\}$ since otherwise this cycle is also present in~$H' - Y'$. Then there must be some~$u \in V(H) \cap (N_G(T)\setminus\{v\})$ on the cycle to leave the tree~$T$ in~$G$, so~$u,v \not\in Y'$. But~$H'$ contains a double-edge between~$u$ and~$v$, so~$H'-Y'$ contains a cycle, contradicting that~$Y'$ is a FVS in~$H'$.
 
 \subparagraph*{FVS-safe} 
 We finally show \cref{op:remove-edge} is FVS-safe. Let~$Z'$ be a minimum FVS in~$G'$, and suppose~$Z'$ is not a FVS in~$G$. Then~$G-Z'$ contains a cycle. This cycle contains the edge~$\{v,w\}$ since otherwise~$G'-Z'$ also contains this cycle. Since~$G'$ contains double-edges between~$v$ and all~$u \in N_G(T)\setminus\{v\}$ and~$v \not\in Z'$, it follows that~$N_G(T)\setminus\{v\} \subseteq Z'$, but then no cycle in~$G-Z'$ can intersect~$T$ and~$\{v,w\}$ is not part of a cycle in~$G-Z'$. We conclude by contradiction that~$Z'$ is a FVS in~$G$. 
 
 To prove optimality, consider a minimum FVS~$Z$ in~$G$ and observe that~$(Z,V(G-Z))$ is an antler in~$G$. Since \cref{op:remove-edge} is antler-safe we know~$G'$ contains an antler~$(C',F')$ with~$C' \cup F' = (Z \cup V(G-Z)) \cap V(G') = V(G')$ and~$|C'| = |Z| - |(Z \cup V(G-Z)) \cap S| = |Z|$. Since~$C' \cup F' = V(G')$ we know~$C'$ is a FVS in~$G'$, hence~$\fvs(G') \leq |C'| = |Z|$, therefore~$|Z'| = \fvs(G') \leq |Z| = \fvs(G)$.
\end{proof}

Before showing one of these reduction rules can be applied efficiently given a reducible FVC, we give the following lemma, which helps in particular with the application of \cref{op:v-flower}.

\begin{lemma}[{Cf.~\cite[Lemma 3.9]{RaymondT17}}]\label{lem:alt-fvs}
If~$v$ is a vertex in a multigraph~$G$ such that~$v$ does not have a self-loop and~$G-v$ is acyclic, then we can find in~$\Oh(n)$ time a set~$X \subseteq V(G)\setminus\{v\}$ such that~$G-X$ is acyclic and~$G$ contains a $v$-flower of order~$|X|$.
\end{lemma}
\begin{proof}
We prove the existence of such a set~$X$ and $v$-flower by induction on~$|V(G)|$. The inductive proof can easily be translated into a linear-time algorithm. If~$G$ is acyclic, output~$X = \emptyset$ and a $v$-flower of order~$0$. Otherwise, since~$v$ does not have a self-loop there is a tree~$T$ of the forest~$G - v$ such that~$G[V(T) \cup \{v\}]$ contains a cycle. Root~$T$ at an arbitrary vertex. For a node~$x \in V(T)$, let~$T_x$ denote the subtree of~$T$ rooted at~$x$. Consider a deepest node~$x$ in~$T$ for which the subgraph~$G[V(T_x) \cup \{v\}]$ contains a cycle~$C$. Then any feedback vertex set of~$G$ that does not contain~$v$, has to contain at least one vertex of~$T_x$. The choice of~$x$ as a deepest vertex implies that~$x$ lies on all cycles of~$G$ that intersect~$T_x$. By induction on~$G' := G - V(T_x)$ and~$v$, there is a feedback vertex set~$X' \subseteq V(G') \setminus \{v\}$ of~$G'$ and a $v$-flower in~$G'$ of order~$|X'|$. We obtain a $v$-flower of order~$|X'| + 1$ in~$G$ by adding~$C$, while~$X := X' \cup \{x\} \subseteq V(G) \setminus \{v\}$ is a feedback vertex set of size~$|X'| + 1$.
\end{proof}

Finally we show that when we are given a reducible FVC~$(C,F)$ in~$G$, then we can find and apply a rule in~$\Oh(n^3)$ time. With a more careful analysis better running time bounds can be shown, but this does not affect the final running time of the main algorithm.

\begin{lemma} \label{lem:fvc-kernel}
 Given a multigraph~$G$ and a reducible FVC~$(C,F)$ in~$G$, we can find and apply a rule in~$\Oh(n^3)$ time.
\end{lemma}
\begin{proof}
 We assume without loss of generality that~$C \subseteq N_G(F)$, since any vertex~$v \in C \setminus N_G(F)$ can be omitted from~$C$ to obtain another FVC satisfying the preconditions to the lemma.
 
 Note that if a vertex~$v \in V(G)$ has a self-loop then~$(\{v\}, \emptyset)$ is an antler in~$G$ and we can apply \cref{op:remove-antler}. If a vertex~$v$ has degree 0 or 1 then~$(\emptyset, \{v\})$ is an antler in~$G$. Hence \Cref{op:edge-multiplicity,,op:contract-edge,,op:remove-antler} can always be applied if the graph contains a self-loop, a vertex with degree less than 3, or more than 2 edges between two vertices. So assume~$G$ is a graph with no self-loops, minimum degree at least 3, and at most two edges between any pair of vertices.

 \cref{def:reducible:fvc} ensures~$|F| > f_r(|C|) = 2|C|^3 + 4|C|^2$ and \cref{lem:degreesum} yields~$e(C,F) \geq |F|$. Hence there must be a vertex~$v$ in~$C$ with more than~$\frac{1}{|C|} \cdot (2|C|^3+4|C|^2) = 2|C|^2+4|C|$ edges to~$F$. By applying \cref{lem:alt-fvs} to~$G[F \cup \{v\}]$ we obtain a set~$X \subseteq F$ such that~$G[F \cup \{v\}] - X$ is acyclic and~$G[F \cup \{v\}]$ contains a $v$-flower of order~$|X|$. Hence if~$|X| \geq |C|+1$ \cref{op:v-flower} can be applied, so assume~$|X| \leq |C|$. 

 The algorithm now carries out a marking scheme to find a tree~$T$ of~$G[F] - X$ to which \cref{op:remove-edge} can be applied. Initially, all trees are unmarked. Note that for each tree~$T$ of~$G[F] - X$ we have~$N_G(T) \subseteq N_G(F) \cup X$. We mark as follows.
 \begin{itemize}
     \item For each~$u \in (N_G(F) \cup X) \setminus \{v\}$, mark up to~$|C|+1$ trees~$T'$ in~$G[F]-X$ for which~$\{u,v\} \in N_G(T')$. 
     
     (That is, for each~$u$ we mark all such trees~$T'$ if there are fewer than~$|C|+1$, while we mark~$|C|+1$ arbitrarily chosen trees~$T'$ with~$\{u,v\} \in N_G(T')$ when there are more.)
 \end{itemize} 
 We prove that there is a tree~$T$ of~$G[F] - X$ that ends up unmarked, which contains a neighbor~$w$ of~$v$ to which \cref{op:remove-edge} can be applied. 
 
 We mark at most~$(|C|+1) \cdot |(N_G(F) \cup X) \setminus \{v\}|$ trees in total. By \cref{def:reducible:fvc}, we have~$|N_G(F)| \leq |C|+1$. Since~$|X| \leq |C|$ and we have~$v \in C \subseteq  N_G(F)$, we therefore mark at most~$(|C|+1) \cdot (|C| + 1 + |C| - 1) = 2|C|^2 + 2|C|$ trees in total over all choices of~$u$. Vertex~$v$ has exactly one edge to each marked tree: there is at least one since~$v \in N_G(T)$ and there cannot be more since~$G[F \cup \{v\}]-X$ is acyclic. In addition,~$v$ has at most~$2|X| \leq 2|C|$ edges to~$X \subseteq F$. Since~$v$ has more than~$2|C|^2 + 4|C|$ edges to~$F$, this means~$v$ has at least one edge to a vertex~$w \in F$ such that~$w \notin X$ and~$w$ does not belong to any marked tree. Let~$T$ be the tree of~$G[F] - X$ that contains~$w$ and observe that~$N_G(T) \subseteq N_G(F) \cup X$.
 
 For each~$u \in N_G(T) \setminus \{v\}$, the tree~$T$ was eligible for marking on account of the pair~$\{u,v\} \in N_G(T)$. Since we did not mark~$T$, there are more than~$|C|$ other trees~$T'$ in~$G[F]-X$ with~$\{u,v\} \subseteq N_G(T')$. This shows that \cref{op:remove-edge} is applicable to~$T$ and~$w$.
 
 It can easily be verified that all operations described can be performed in~$\Oh(n^3)$ time.
%
\end{proof}

\section{Finding and removing antlers} \label{sec:find-antler}
We will find antlers making use of color coding, using coloring functions of the form~$\chi \colon V(G) \cup E(G) \to \{\cF,\cC,\cR\}$ that assign colors to both vertices and edges. For all~$c \in \{\cF, \cC, \cR\}$ we define the inverse of the coloring, restricted to the vertices, as~$\inv{\chi}_V(c) = \inv{\chi}(c) \cap V(G)$.
For any integer~$z\geq0$, a $z$-antler~$(C,F)$ in a multigraph~$G$ is \emph{$z$-properly colored} by a coloring~$\chi$ if all of the following hold:
\begin{enumerate}[label=(\roman*)]
 \item $F \subseteq \inv{\chi}_V(\cF)$,
 \item $C \subseteq \inv{\chi}_V(\cC)$,\label[property]{prop:c:colored}
 \item $N_G(F)\setminus C \subseteq \inv{\chi}_V(\cR)$, and
 \item $G[C \cup F] - \inv{\chi}(\cR)$ is a $C$-certificate of order~$z$.
\end{enumerate}
By this definition, whether or not a $z$-antler~$(C,F)$ is $z$-properly colored depends only on the colors of the \emph{vertices}~$C \cup F \cup N_G(F)$ and on the colors of the \emph{edges} of~$G[C\cup F]$.

Recall that~$\inv{\chi}(\cR)$ can contain edges as well as vertices so for any subgraph~$H$ of~$G$ the multigraph~$H - \inv{\chi}(\cR)$ is obtained from~$H$ by removing both vertices and edges. It can be seen that if~$(C,F)$ is a $z$-antler, then there exists a coloring that $z$-properly colors it. Consider for example a coloring where a vertex~$v$ is colored~\cC (resp.~\cF) if~$v \in C$ (resp.~$v \in F$), all other vertices are colored~\cR, and for some $C$-certificate~$H$ of order~$z$ in~$G[C \cup F]$ all edges in~$H$ have color~\cF and all other edges have color~\cR. 
The property of a properly colored $z$-antler described in \cref{lem:propcolor} will be useful to prove correctness of the color coding algorithm.

\begin{lemma} \label{lem:propcolor}
 For any~$z\geq0$, if a $z$-antler~$(C,F)$ in multigraph~$G$ is $z$-properly colored by a coloring~$\chi$ and~$H$ is a component of~$G[C \cup F] - \inv{\chi}(\cR)$, then each component~$H'$ of~$H-C$ is a component of~$G[\inv{\chi}_V(\cF)] - \inv{\chi}(\cR)$ with~$N_{G-\inv{\chi}(\cR)}(H') \subseteq C \cap V(H)$.
\end{lemma}
\begin{proof}
 By \cref{prop:c:colored} we have~$C \cap \inv{\chi}_V(\cF) = \emptyset$, so~$N_{G-\inv{\chi}(\cR)}(H') \subseteq C \cap V(H)$ implies that~$N_{G[\inv{\chi}_V(\cF)]-\inv{\chi}(\cR)}(H') = \emptyset$ and hence that~$H'$ is a component of~$G[\inv{\chi}_V(\cF)] - \inv{\chi}(\cR)$. We show~$N_{G-\inv{\chi}(\cR)}(H') \subseteq C \cap V(H)$.
 
 Suppose~$v \in N_{G-\inv{\chi}(\cR)}(H')$ and let~$u \in V(H')$ be a neighbor of~$v$ in~$G-\inv{\chi}(\cR)$. Since~$V(H') \subseteq F$ we know~$u \in F$. Since~$(C,F)$ is $z$-properly colored we also have~$N_G(F)\setminus C = \inv{\chi}(\cR)$, hence~$N_G(u) \subseteq C \cup F \cup \inv{\chi}(\cR)$ so then~$N_{G-\inv{\chi}(\cR)}(u) \subseteq C \cup F$. By choice of~$u$ we have~$v \in N_{G-\inv{\chi}(\cR)}(u) \subseteq C \cup F$. So since~$u,v \in C \cup F$, and~$u$ and~$v$ are neighbors in~$G-\inv{\chi}(\cR)$ we know~$u$ and~$v$ are in the same component of~$G[C \cup F]-\inv{\chi}(\cR)$, hence~$v \in V(H)$.
 
 Suppose~$v \not\in C$, so~$v \in F$. Since also~$u \in F$ we know that~$u$ and~$v$ are in the same component of~$G[F]-\inv{\chi}(\cR)$, so~$v \in H'$, but then~$v \not\in N_{G-\inv{\chi}(\cR)}(H')$ contradicting our choice of~$v$. It follows that~$v \in C$ hence~$v \in C \cap V(H)$. Since~$v \in N_{G-\inv{\chi}(\cR)}(H')$ was chosen arbitrarily, $N_{G-\inv{\chi}(\cR)}(H') \subseteq C \cap V(H)$.
\end{proof}

We now show that a $z$-antler can be obtained from a suitable coloring~$\chi$ of the graph. The algorithm we give updates the coloring~$\chi$ and recolors any vertex or edge that is not part of a $z$-properly colored antler to color~\cR. We show that after repeatedly refining the coloring, the coloring that we arrive at identifies a suitable antler.

\begin{lemma}\label{lem:alg}
There exists an $n^{\Oh(z)}$-time algorithm taking as input an integer~$z\geq0$, a multigraph~$G$, and a coloring~$\chi$ and producing as output a (possibly empty) $z$-antler~$(C,F)$ in~$G$, such that for any $z$-antler~$(\hat{C},\hat{F})$ that is $z$-properly colored by~$\chi$ we have~$\hat{C} \subseteq C$ and~$\hat{F} \subseteq F$.
\end{lemma}
\begin{proof}
 We define a function~$W_\chi \colon 2^{\inv{\chi}_V(\cC)} \to 2^{\inv{\chi}_V(\cF)}$ as follows: for any~$C \subseteq \inv{\chi}_V(\cC)$ let~$W_\chi(C)$ denote the set of all vertices that are in a component~$H$ of~$G[\inv{\chi}_V(\cF)] - \inv{\chi}(\cR)$ for which~$N_{G - \inv{\chi}(\cR)}(H) \subseteq C$. Intuitively, the function~$W_\chi$ maps sets~$C$ of~$\cC$-colored vertices to those~$\cF$-colored vertices~$F$ for which~$F$ could be used to build a~$C$-certificate.

 The algorithm below updates the coloring~$\chi$ and recolors any vertex or edge that is not part of a $z$-properly colored antler to color~\cR.
 \begin{enumerate}
  \item \label[step]{antleralg:init}  Recolor all edges to color~\cR when one of its endpoints has color~\cR.
  \item \label[step]{antleralg:loop1} For each component~$H$ of~$G[\inv{\chi}_V(\cF)]$ we recolor all vertices of~$H$ and their incident edges to color~\cR if~$e(H,\inv{\chi}_V(\cR)) > 1$ or~$H$ is not a tree.
  \item \label[step]{antleralg:loop2} For each subset~$C \subseteq \inv{\chi}_V(\cC)$ of size at most~$z$, mark all vertices in~$C$ if~$\fvs(G[C \cup W_\chi(C)] - \inv{\chi}(\cR)) = |C|$.
  \item \label[step]{antleralg:repeat} If~$\inv{\chi}_V(\cC)$ contains unmarked vertices we recolor them to color~\cR, remove markings made in \cref{antleralg:loop2} and repeat from \cref{antleralg:init}.
  \item \label[step]{antleralg:return} If all vertices in~$\inv{\chi}_V(\cC)$ are marked in \cref{antleralg:loop2}, return the antler~$(\inv{\chi}_V(\cC),\inv{\chi}_V(\cF))$.
 \end{enumerate}
 
 \subparagraph*{Running time}
 The algorithm will terminate after at most~$n$ iterations since in every iteration the number of vertices in~$\inv{\chi}_V(\cR)$ increases. \Cref{antleralg:init,antleralg:loop1,antleralg:repeat,antleralg:return} can easily be seen to take no more than~$\Oh(n^2)$ time. \Cref{antleralg:loop2} can be performed in~$\Oh(4^z \cdot n^{z+1})$ time by checking for all~$\Oh(n^z)$ subsets~$C \in \inv{\chi}_V(\cC)$ of size at most~$z$ whether the subgraph~$G[C \cup W_\chi(C)] - \inv{\chi}(\cR)$ has feedback vertex number~$z$. This can be done in time~$\Oh(4^z \cdot n)$~\cite{IwataK21}.
 Hence the algorithm runs in time~$n^{\Oh(z)}$.
 
 \subparagraph*{Correctness}
 We show first that any $z$-properly colored antler prior to executing the algorithm remains $z$-properly colored after termination. Afterwards we argue that in \cref{antleralg:return}, the pair~$(\inv{\chi}_V(\cC),\inv{\chi}_V(\cF))$ is a $z$-antler in~$G$. Since~$(\inv{\chi}_V(\cC),\inv{\chi}_V(\cF))$ contains all properly colored antlers this proves correctness.
 
 \begin{claim}\label{claim:remain}
  All $z$-antlers~$(\hat{C},\hat{F})$ that are $z$-properly colored by~$\chi$ prior to executing the algorithm are also $z$-properly colored by~$\chi$ after termination of the algorithm.
 \end{claim}
 \begin{claimproof}
  To show the algorithm preserves properness of the coloring, we show that every individual recoloring preserves properness, that is, if an arbitrary $z$-antler is $z$-properly colored prior to the recoloring, it is also $z$-properly colored after the recoloring.
  
  Suppose an arbitrary $z$-antler~$(\hat{C},\hat{F})$ is $z$-properly colored by~$\chi$. An edge is only recolored when one of its endpoints has color~\cR. Since these edges are not in~$G[\hat{C} \cup \hat{F}]$, their color does not change whether~$(\hat{C},\hat{F})$ is colored $z$-properly. All other operations done by the algorithm are recolorings of vertices to color~\cR. We show that any time a vertex~$v$ is recolored we have that~$v \not\in \hat{C} \cup \hat{F}$, meaning~$(\hat{C},\hat{F})$ remains colored $z$-properly.
  
  Suppose~$v$ is recolored in \cref{antleralg:loop1}, then we know~$\chi(v) = \cF$, and~$v$ is part of a component~$H$ of~$G[\inv{\chi}_V(\cF)]$. Since~$\chi$ $z$-properly colors~$(\hat{C},\hat{F})$ we have~$\hat{F} \subseteq \inv{\chi}_V(\cF)$ but~$N_G(\hat{F}) \cap \inv{\chi}_V(\cF) = \emptyset$, so since~$H$ is a component of~$G[\inv{\chi}_V(\cF)]$ we know either~$V(H) \subseteq \hat{F}$ or~$V(H) \cap \hat{F} = \emptyset$. If~$V(H) \cap \hat{F} = \emptyset$ then clearly~$v \not\in \hat{C} \cup \hat{F}$. So suppose~$V(H) \subseteq \hat{F}$, then~$H$ is a tree in~$G[\hat{F}]$. Since~$v$ was recolored and~$H$ is a tree it must be that~$e(H,\inv{\chi}_C(\cR)) > 1$ but this contradicts that~$(\hat{C},\hat{F})$ is a FVC.
  
  Suppose~$v$ is recolored in \cref{antleralg:repeat}, then we know~$v$ was not marked during \cref{antleralg:loop2} and~$\chi(v) = \cC$, so~$v \not\in \hat{F}$. Suppose that~$v \in \hat{C}$. We derive a contradiction by showing that~$v$ was marked in \cref{antleralg:loop2}.
  Since~$(\hat{C},\hat{F})$ is $z$-properly colored, we know that~$G[\hat{C} \cup \hat{F}] - \inv{\chi}(\cR)$ is a $\hat{C}$-certificate of order~$z$, so if~$H$ is the component of~$G[\hat{C} \cup \hat{F}] - \inv{\chi}(\cR)$ containing~$v$ then~$\fvs(H) = |\hat{C} \cap V(H)| \leq z$. Since~$\hat{C} \cap V(H) \subseteq \hat{C} \subseteq \inv{\chi}_V(\cC)$ we know that in some iteration in \cref{antleralg:loop2} we have~$C = \hat{C} \cap V(H)$.
  To show that~$v$ was marked, we show that~$\fvs(G[C \cup W_\chi(C))] - \inv{\chi}(\cR)) = |C|$. We know~$G[W_\chi(C)] - \inv{\chi}(\cR)$ is a forest since it is a subgraph of~$G[\inv{\chi}_V(\cF)]$ which is a forest by \cref{antleralg:loop1}, so we have that~$\fvs(G[C \cup W_\chi(C)] - \inv{\chi}(\cR)) \leq |C|$.
  To show~$\fvs(G[C \cup W_\chi(C))] - \inv{\chi}(\cR)) \geq |C|$ we show that~$H$ is a subgraph of~$G[C \cup W_\chi(C))] - \inv{\chi}(\cR)$. By \cref{lem:propcolor} we have that each component~$H'$ of~$H-\hat{C}$ is also a component of~$G[\inv{\chi}_V(\cF)] - \inv{\chi}(\cR)$ with~$N_{G-\inv{\chi}(\cR)}(H') \subseteq \hat{C} \cap V(H) = C$. Hence~$V(H-\hat{C}) = V(H-C) \subseteq W_\chi(C)$ so~$H$ is a subgraph of~$G[C \cup W_\chi(C)]$. Since~$H$ is also a subgraph of~$G[\hat{C} \cup \hat{F}] - \inv{\chi}(\cR)$ we conclude that~$H$ is a subgraph of~$G[C \cup W_\chi(C))] - \inv{\chi}(\cR)$ and therefore~$\fvs(G[C \cup W_\chi(C))] - \inv{\chi}(\cR)) \geq \fvs(H) = |C|$.
 \end{claimproof}
 
 \begin{claim}\label{claim:coloredantler}
  In \cref{antleralg:return},~$(\inv{\chi}_V(\cC),\inv{\chi}_V(\cF))$ is a $z$-antler in~$G$.
 \end{claim}
 \begin{claimproof}
  We know~$(\inv{\chi}_V(\cC),\inv{\chi}_V(\cF))$ is a FVC in~$G$ because each connected component of~$G[\inv{\chi}_V(\cF)]$ is a tree and has at most one edge to a vertex not in~$\inv{\chi}_V(\cC)$ by \cref{antleralg:loop1}. It remains to show that~$G[\inv{\chi}_V(\cC) \cup \inv{\chi}_V(\cF)]$ contains a $\inv{\chi}_V(\cC)$-certificate of order~$z$.
  Note that in \cref{antleralg:return} the coloring~$\chi$ is the same as in the last execution of \cref{antleralg:loop2}. Let~$\mathcal{C} \subseteq 2^{\inv{\chi}_V(\cC)}$ be the family of all subsets~$C \subseteq \inv{\chi}_V(\cC)$ that have been considered in \cref{antleralg:loop2} and met the conditions for marking all vertices in~$C$, i.e.,~$\fvs(G[C \cup W_\chi(C)] - \inv{\chi}(\cR)) = |C| \leq z$. Since all vertices in~$\inv{\chi}_V(\cC)$ have been marked during the last execution of \cref{antleralg:loop2} we know~$\bigcup_{C \in \mathcal{C}} C = \inv{\chi}_V(\cC)$.
  
  Let~$C_1,\ldots,C_{|\mathcal{C}|}$ be the sets in~$\mathcal{C}$ in an arbitrary order and define~$D_i := C_i \setminus C_{<i}$ for all~$1 \leq i \leq |\mathcal{C}|$. Observe that~$D_1, \ldots, D_{|\mathcal{C}|}$ is a partition of~$\inv{\chi}_V(\cC)$ with~$|D_i| \leq z$ and~$C_i \subseteq D_{\leq i}$ for all~$1 \leq i \leq |\mathcal{C}|$. Note that~$D_i$ may be empty for some~$i$.
  
  Using these definitions we now show that~$G[\inv{\chi}_V(\cC) \cup \inv{\chi}_V(\cF)]$ contains a $\inv{\chi}_V(\cC)$-certificate of order~$z$. We do this by showing there are~$|\mathcal{C}|$ vertex disjoint subgraphs of~$G[\inv{\chi}_V(\cC) \cup \inv{\chi}_V(\cF)]$, call them~$G_1, \ldots, G_{|\mathcal{C}|}$, such that~$\fvs(G_i) = |D_i| \leq z$ for each~$1 \leq i \leq |\mathcal{C}|$.
  Take~$G_i := G[D_i \cup (W_\chi(D_{\leq i}) \setminus W_\chi(D_{<i}))] - \inv{\chi}(\cR)$ for all~$1 \leq i \leq |\mathcal{C}|$. First we show that for any~$i \neq j$ the multigraphs~$G_i$ and~$G_j$ are vertex disjoint. Clearly~$D_i \cap D_j = \emptyset$. We can assume~$i < j$, so~$D_{\leq i} \subseteq D_{<j}$ and then~$W_\chi(D_{\leq i}) \subseteq W_\chi(D_{<j})$. Observe that 
  \begin{align*}
   (W_\chi(D_{\leq i}) \setminus W_\chi(D_{<i})) \cap (W_\chi(D_{\leq j}) \setminus W_\chi(D_{<j})) = \emptyset,
  \end{align*}
  which follows from the fact that the expression essentially says~$(P \setminus Q) \cap (R \setminus S) = \emptyset$, which holds since we derived above that~$P \subseteq S$.
  
  We complete the proof by showing~$\fvs(G_i) = |D_i|$ for all~$1 \leq i \leq \ell$. Recall that~$D_i = C_i \setminus C_{<i}$. Since~$C_i \in \mathcal{C}$ we know~$C_i$ is an optimal FVS in~$G[C_i \cup W_\chi(C_i)] - \inv{\chi}(\cR)$, so then clearly~$D_i$ is an optimal FVS in~$G[C_i \cup W_\chi(C_i)] - \inv{\chi}(\cR) - C_{<i} = G[D_i \cup W_\chi(C_i)] - \inv{\chi}(\cR)$. We know that~$C_i \subseteq D_{\leq i}$ so then also~$W_\chi(C_i) \subseteq W_\chi(D_{\leq i})$. It follows that~$D_i$ is an optimal FVS in~$G[D_i \cup W_\chi(D_{\leq i})] - \inv{\chi}(\cR)$. In this multigraph, all vertices in~$W_\chi(D_{<i})$ must be in a component that does not contain any vertices from~$D_i$, so this component is a tree and we obtain~$|D_i| = \fvs(G[D_i \cup W_\chi(D_{\leq i})] - \inv{\chi}(\cR)) = \fvs(G[D_i \cup W_\chi(D_{\leq i})] - \inv{\chi}(\cR) - W_\chi(D_{<i})) = \fvs(G[D_i \cup (W_\chi(D_{\leq i}) \setminus W_\chi(D_{<i}))] - \inv{\chi}(\cR)) = \fvs(G_i)$. 
 \end{claimproof}
 
 It can be seen from \cref{claim:remain} that for any $z$-properly colored antler~$(\hat{C},\hat{F})$ we have~$\hat{C} \subseteq \inv{\chi}_V(\cC)$ and~$\hat{F} \subseteq \inv{\chi}_V(\cF)$. \Cref{claim:coloredantler} completes the correctness argument.
\end{proof}

If a multigraph~$G$ contains a reducible single-tree FVC of width at most~$k$ then we can find and apply a reduction rule by \cref{lem:find-fvc} and \cref{lem:fvc-kernel}. If~$G$ does not contain such a FVC, but~$G$ does contain a non-empty $z$-antler~$(C,F)$ of width at most~$k$, then using \cref{col:few-fvcs} we can prove that whether~$(C,F)$ is $z$-properly colored is determined by the color of~$\Oh(k^5 z^2)$ relevant vertices and edges.
Similar to the algorithm from \cref{lem:find-fvc}, we can use two $(n+m,\Oh(k^5 z^2))$-universal sets to create a set of colorings that is guaranteed to contain a coloring that $z$-properly colors~$(C,F)$.
Using \cref{lem:alg} we find a non-empty $z$-antler and apply \cref{op:remove-antler}. We obtain the following:
\begin{lemma}\label{lem:find-op}
 Consider a multigraph~$G$ and integers~$k \geq z \geq 0$. If~$G$ contains a non-empty $z$-antler of width at most~$k$, we can find and apply a reduction rule in~$2^{\Oh(k^5 z^2)} \cdot n^{\Oh(z)}$ time.
\end{lemma}
\begin{proof}
 Consider the following algorithm:
 First we use \cref{lem:find-fvc} to obtain a FVC~$(C_1,F_1)$ in~$2^{\Oh(k^3)}\cdot n^3 \log n$ time. If~$(C_1,F_1)$ is reducible we can find and apply a reduction rule in~$\Oh(n^3)$ time by \cref{lem:fvc-kernel}, so assume~$(C_1,F_1)$ is not reducible; this implies~$G$ does not contain a reducible single-tree FVC of width at most~$k$. Create two $(n+m, k+k^2 + 60k^5z^2)$-universal sets~$\mathcal{U}_1$ and~$\mathcal{U}_2$ for~$V(G) \cup E(G)$ using \cref{thm:universal-set}. Define for each pair~$(Q_1,Q_2) \in \mathcal{U}_1 \times \mathcal{U}_2$ the coloring~$\chi_{Q_1,Q_2}$ of~$G$ that assigns all vertices and edges in~$Q_1$ color~\cC, all vertices and edges in~$Q_2 \setminus Q_1$ color~\cF, and all vertices and edges not in~$Q_1 \cup Q_2$ color~\cR. For each~$(Q_1,Q_2) \in \mathcal{U}_1 \times \mathcal{U}_2$ obtain in~$n^{\Oh(z)}$ time a $z$-antler~$(C_2,F_2)$ by running the algorithm from \cref{lem:alg} on~$G$ and~$\chi_{Q_1,Q_2}$. If~$(C_2, F_2)$ is not empty, apply \cref{op:remove-antler} to remove~$(C_2,F_2)$, otherwise report~$G$ does not contain a $z$-antler of width at most~$k$.
 
 \subparagraph*{Running time}
 By \cref{thm:universal-set}, the sets~$\mathcal{U}_1$ and~$\mathcal{U}_2$ have size~$2^{\Oh(k^5 z^2)} \cdot \log n$ and can be created in~$2^{\Oh(k^5 z^2)} \cdot n \log n$ time. It follows that there are~$|\mathcal{U}_1 \times \mathcal{U}_2| = 2^{\Oh(k^5 z^2)} \cdot \log^2 n$ colorings for which we apply the~$n^{\Oh(z)}$ time algorithm from \cref{lem:alg}. We obtain an overall running time of~$2^{\Oh(k^5 z^2)} \cdot n^{\Oh(z)}$. 
 
 \subparagraph*{Correctness}
 Suppose~$G$ contains a $z$-antler~$(C,F)$ of width at most~$k$. We show that the algorithm finds a reduction rule to apply.
 By \cref{col:few-fvcs} we know that there exists an~$F' \subseteq F$ such that~$(C,F')$ is a $z$-antler where~$G[F']$ has at most~$\frac{|C|}{2}(z^2+4z-1)$ trees. For each tree~$T$ in~$G[F']$ note that~$(C,V(T))$ is a single-tree FVC of width~$|C| \leq k$. If for some tree~$T$ in~$G$ the FVC~$(C,V(T))$ is reducible, then the FVC~$(C_1,F_1)$ obtained by the algorithm via \cref{lem:find-fvc} is guaranteed to be reducible and we find a reduction rule using \cref{lem:fvc-kernel}. So suppose for all trees~$T$ in~$G[F']$ that~$|V(T)| \leq f_r(|C|)$; recall that the condition~$|N_G(F)| \leq |C|+1$ of \cref{def:reducible:fvc} trivially holds for a single-tree FVC. Then~$|F'| \leq \frac{|C|}{2}(z^2+4z-1) \cdot f_r(|C|)$. We show that in this case there exists a pair~$(Q_1,Q_2) \in \mathcal{U}_1 \times \mathcal{U}_2$ such that~$\chi_{Q_1,Q_2}$ $z$-properly colors~$(C,F')$.
 
 Whether a coloring $z$-properly colors~$(C,F')$ is only determined by the colors of~$C \cup F' \cup N_G(F') \cup E(G[C \cup F'])$.
 
 \begin{claim} \label{claim:antler-color-size}
 ~$|C \cup F' \cup N_G(F') \cup E(G[C \cup F'])| \leq k + k^2 + 60 k^5 z^2$.
 \end{claim}
 \begin{claimproof}
  Note that~$|N_G(F') \setminus C| \leq \frac{|C|}{2}(z^2+4z-1)$ since no tree in~$G[F']$ can have more than one neighbor outside~$C$. Additionally we have
  \allowdisplaybreaks
  \begin{align*}
    |E(G[C \cup F'])|
    &\leq |E(G[C])| + |E(G[F'])| + e(C,F') \\
    &\leq |E(G[C])| + |F'| + |C| \cdot |F'|          &\hspace{-5em}\text{since~$G[F']$ is a forest} \\
    &\leq |C|^2 + (|C|+1) \cdot |F'|\\
    &\leq |C|^2 + (|C|+1) \cdot \frac{|C|}{2}(z^2+4z-1) \cdot f_r(|C|)\\
    &\leq k^2 + (k+1) \cdot \frac{k}{2}(z^2+4z-1) \cdot (2k^3 + 4k^2)\\
    &\leq k^2 + \frac{z^2+4z-1}{2} \cdot (k^2+k) \cdot (2k^3 + 4k^2)\\
    &\leq k^2 + \frac{z^2+4z-1}{2} \cdot 2k^2 \cdot 6k^3 &\hspace{-5em}\text{since~$k \leq k^2$}\\
    &\leq k^2 + \frac{5z^2}{2} \cdot 12k^5            &\hspace{-5em}\text{since~$z=0$ or $z \geq 1$}\\
    &\leq k^2 + 30 k^5 z^2,
  \end{align*}
  hence
  \begin{align*}
    |C \cup F' \cup N_G(F') \cup E(G[C \cup F'])|\\
  &\hspace{-10em}= |C| + |F'| + |N_G(F')\setminus C| + |E(G[C \cup F'])|\\
  &\hspace{-10em}\leq |C|
    + \frac{|C|}{2}(z^2+4z-1) \cdot f_r(|C|) 
    + \frac{|C|}{2}(z^2+4z-1)
    + k^2 + 30 k^5 z^2 \\
  &\hspace{-10em}= |C|
    + \frac{|C|}{2}(z^2+4z-1) \cdot (f_r(|C|)+1) 
    + k^2 + 30 k^5 z^2 \\
  &\hspace{-10em}\leq k
    + \frac{k}{2}(5z^2) \cdot (f_r(k)+1) 
    + k^2 + 30 k^5 z^2 \\
  &\hspace{-10em}\leq k
    + \frac{5}{2}(k z^2) \cdot (2k^3 + 4k^2 + 1)
    + k^2 + 30 k^5 z^2 \\
  &\hspace{-10em} \leq k+k^2 + 60k^5z^2,
  \end{align*}
  where the last step follows from an elementary computation.
 \end{claimproof}
 
  By construction of~$\mathcal{U}_1$ and~$\mathcal{U}_2$ there exist sets~$Q_1 \in \mathcal{U}_1$ and~$Q_2 \in \mathcal{U}_2$ such that~$\chi_{Q_1,Q_2}$ $z$-properly colors~$(C,F')$. Hence the algorithm from \cref{lem:alg} returns a non-empty $z$-antler for~$\chi_{Q_1,Q_2}$ and \cref{op:remove-antler} can be executed.
 \end{proof}

Note that applying a reduction rule reduces the number of vertices, removes an edge between a pair of vertices that is not connected by a double-edge, or increases the number of double-edges. Hence by repeatedly using \cref{lem:find-op} to apply a reduction rule we obtain, after at most~$\Oh(n^2)$ iterations, a graph in which no rule applies. By \cref{lem:find-op} this graph does not contain a non-empty $z$-antler of width at most~$k$. In the remainder of this section, we show that this method reduces the solution size at least as much as iteratively removing $z$-antlers of width at most~$k$. We first describe the behavior of such a sequence of antlers.

For integers~$k\geq0$ and~$z\geq0$, we say a sequence of disjoint vertex sets~$C_1, F_1,\ldots,C_\ell, F_\ell$ is a \emph{$z$-antler-sequence} for a multigraph~$G$ if for all~$1 \leq i \leq \ell$ the pair~$(C_i, F_i)$ is a $z$-antler in~$G - (C_{<i} \cup F_{<i})$. The \emph{width} of a $z$-antler-sequence is defined as~$\max_{1 \leq i \leq \ell} |C_i|$.

\begin{proposition} \label{prop:antler-sequence}
 If~$C_1, F_1,\ldots,C_\ell, F_\ell$ is a \emph{$z$-antler-sequence} for some multigraph~$G$, then the pair~$(C_{\leq i}, F_{\leq i})$ is a $z$-antler in~$G$ for any~$1 \leq i \leq \ell$.
\end{proposition}
\begin{proof}
 We use induction on~$i$. Clearly the statement holds for~$i = 1$, so suppose~$i > 1$. By induction~$(C_{<i},F_{<i})$ is a $z$-antler in~$G$, and since~$(C_i,F_i)$ is a $z$-antler in~$G-(C_{<i} \cup F_{<i})$ we have by \cref{lem:antler_combine} that~$(C_{<i} \cup C_i,F_{<i} \cup F_i) = (C_{\leq i},F_{\leq i})$ is a $z$-antler in~$G$.
\end{proof}

The following theorem describes that repeatedly applying \cref{lem:find-op} reduces the solution size at least as much as repeatedly removing $z$-antlers of width at most~$k$. By taking~$t=1$ we obtain \cref{thm:z:antler}.
\begin{theorem}\label{thm:main}
 Given as input a multigraph~$G$ and integers~$k\geq z \geq 0$ we can find in~$2^{\Oh(k^5 z^2)} \cdot n^{\Oh(z)}$ time a vertex set~$S^* \subseteq V(G)$ such that:
 \begin{enumerate}
  \item \label[condition]{thm:item:fvs-safe} there is a minimum FVS in~$G$ containing all vertices of~$S^*$, and
  \item \label[condition]{thm:item:antlers} for any $z$-antler sequence~$C_1, F_1, \ldots, C_t, F_t$ in~$G$ of width at most~$k$,
  we have~$|S^*| \geq |C_{\leq t}|$.
 \end{enumerate}
\end{theorem}
\begin{proof}
 We first describe the algorithm.
 
 \subparagraph*{Algorithm}
 We invoke \cref{lem:find-op} on~$G$ with~$k$ and~$z$. If no applicable rule was found, then return an empty vertex set~$S^* := \emptyset$. Otherwise, consider the multigraph~$G'$ and vertex set~$S$ obtained from the reduction rule. We recursively call our algorithm on~$G'$ with integers~$z$ and~$k$ to obtain a vertex set~$S'$ and return the vertex set~$S^* := S \cup S'$.
 
 \subparagraph*{Running time}
 Since every reduction rule reduces the number of vertices, removes an edge between a pair of vertices that is not connected by a double-edge, or increases the number of double-edges, after applying~$\Oh(n^2)$ reduction rules we obtain a graph where no rules can be applied. Therefore after~$\Oh(n^2)$ recursive calls the algorithm terminates. We obtain a running time of~$2^{\Oh(k^5 z^2)} \cdot n^{\Oh(z)}$.
 
 \subparagraph*{Correctness}
 We prove correctness by induction on the recursion depth, which is shown to be finite by the running time analysis.
 
 First consider the case that no reduction rule was found. Clearly \cref{thm:item:fvs-safe} holds for~$S^* := \emptyset$. To show \cref{thm:item:antlers}, suppose~$C_1, F_1, \ldots, C_t, F_t$ is a $z$-antler-sequence of width at most~$k$ for~$G$. The first non-empty antler in this sequence is a $z$-antler of width at most~$k$ in~$G$. Since no reduction rule was found using \cref{lem:find-op} it follows that~$G$ does not contain a non-empty $z$-antler of width at most~$k$. Hence all antlers in the sequence must be empty and~$|C_{\leq t}| = 0$, so \cref{thm:item:antlers} holds for~$S^* = \emptyset$.
 
 For the other case, suppose~$G'$ and~$S$ are obtained by applying a reduction rule. Since all rules are FVS-safe, we know for any minimum FVS~$S''$ of~$G'$ that~$S^* := S \cup S''$ is a minimum FVS in~$G$. Since~$S'$ is obtained from a recursive call of strictly smaller recursion depth, by induction there is a minimum FVS in~$G'$ containing all vertices of~$S'$. Let~$S'' \supseteq S'$ be such a FVS in~$G'$; then~$S \cup S''$ is a minimum FVS in~$G$. It follows that there is a minimum FVS in~$G$ containing all vertices of~$S \cup S'$, proving \cref{thm:item:fvs-safe}.
 
 To prove \cref{thm:item:antlers}, suppose~$C_1, F_1, \ldots, C_t, F_t$ is a $z$-antler-sequence of width at most~$k$ for~$G$. We first prove the following:
 
 \begin{claim} \label{claim:antler-sequence}
  There exists a $z$-antler-sequence~$C_1', F_1', \ldots, C_t', F_t'$ of width at most~$k$ for~$G'$
  such that:
  \begin{enumerate}
   \item \label[condition]{item:vsets}~$C_{\leq t}' \cup F_{\leq t}' = V(G') \cap (C_{\leq t} \cup F_{\leq t})$, and
   \item \label[condition]{item:cutsize}~$|C_{\leq t}'| = \sum_{1 \leq i \leq t} (|C_i| - |(C_i \cup F_i) \cap S|)$.
  \end{enumerate}
 \end{claim}
 \begin{claimproof}
  We use induction on~$t$. Since~$G'$ and~$S$ are obtained through an antler-safe reduction rule and~$(C_1,F_1)$ is a $z$-antler in~$G$, by \cref{def:safe} we know that~$G'$ contains a $z$-antler~$(C_1',F_1')$ such that~$C_1' \cup F_1' = (C_1 \cup F_1) \cap V(G')$ and~$|C_1'| = |C_1| - |(C_1 \cup F_1) \cap S|$. The claim holds for~$t=1$.
  
  For the induction step, consider~$t > 1$. By applying induction to the prefix consisting of the first $t-1$ pairs of the sequence, there is a $z$-antler sequence~$C_1',F_1', \ldots, C_{t-1}',F_{t-1}'$ of width at most~$k$ for~$G'$ such that both conditions hold.
  \cref{prop:antler-sequence} ensures that~$(C_{\leq t}, F_{\leq t})$ is a $z$-antler in~$G$. Since~$G'$ and~$S$ are obtained through an antler-safe reduction rule from~$G$ there is a $z$-antler~$(C',F')$ in~$G'$ such that~$C' \cup F' = V(G') \cap (C_{\leq t} \cup F_{\leq t})$ and~$|C'| = |C_{\leq t}| - |S \cap (C_{\leq t} \cup F_{\leq t})|$. Take~$C_{t}' := C' \setminus (C_{<t}' \cup F_{<t}')$ and~$F_{t}' := F' \setminus (C_{<t}' \cup F_{<t}')$. By \cref{lem:antler_diff} the pair~$(C_{t}',F_{t}')$ is a $z$-antler in~$G'-(C_{<t}' \cup F_{<t}')$. It follows that~$C_1',F_1',\ldots,C_{t}',F_{t}'$ is a $z$-antler-sequence for~$G'$. We first show \cref{item:vsets}.
  \begin{align*}
   C_{\leq t}' \cup  F_{\leq t}'
   &= C_{t}' \cup F_{t}' \cup C_{<t}' \cup F_{<t}' \\
   &= C' \cup F'
    &&\text{by choice of~$C_{t}'$ and $F_{t}'$}\\
   &= V(G') \cap (C_{\leq t} \cup F_{\leq t})
    &&\text{by choice of~$C'$ and~$F'$.}
  \end{align*}
  To prove \cref{item:cutsize}, and that the $z$-antler-sequence~$C_1',F_1',\ldots,C_{t}',F_{t}'$ has width at most~$k$, we first show~$|C_{t}'| = |C_t| - |(C_t \cup F_t) \cap S|$. Observe that~$(C_{\leq t}',F_{\leq t}')$ is an antler in~$G'$ by \cref{prop:antler-sequence}.
  \allowdisplaybreaks
  \begin{align*}
   |C_{t}'|
   &= |C_{\leq t}'| - |C_{<t}'|
    &&\hspace{-9em}\text{since~$C_i' \cap C_j' = \emptyset$ for all $i \neq j$}\\
   &= \fvs(G'[C_{\leq t}' \cup F_{\leq t}']) - |C_{<t}'|
    &&\hspace{-9em}\text{by the above} \\
   &= \fvs(G'[V(G') \cap (C_{\leq t} \cup F_{\leq t})]) - |C_{<t}'|
    &&\hspace{-9em}\text{by \cref{item:vsets}} \\
   &= \fvs(G'[C' \cup F']) - |C_{<t}'|
    &&\hspace{-9em}\text{by choice of~$C'$ and $F'$}\\
   &= |C'| - |C_{<t}'|
    &&\hspace{-9em}\text{since~$(C',F')$ is an antler in $G'$} \\
   &= |C'| - \sum_{1 \leq i < t} (|C_i| - |(C_i \cup F_i) \cap S|)
    &&\hspace{-9em}\text{by induction}\\
   &= |C_{\leq t}| - | S \cap (C_{\leq t} \cup F_{\leq t})|
    - \sum_{1 \leq i < t} (|C_i| - |(C_i \cup F_i) \cap S|)\\
   &= \sum_{1 \leq i \leq t} (|C_i| - |S \cap (C_i \cup F_i)|)
    - \sum_{1 \leq i    < t} (|C_i| - |(C_i \cup F_i) \cap S|)\\
    &&&\hspace{-14em}\text{since~$C_1,F_1,\ldots,C_t,F_t$ are pairwise disjoint}\\
   &= |C_{t}| - |(C_{t} \cup F_{t}) \cap S|
  \end{align*}
  
  We know the $z$-antler-sequence~$C_1',F_1',\ldots,C'_{t-1},F'_{t-1}$ has width at most~$k$, so to show that this $z$-antler-sequence has width at most~$k$ it suffices to prove that~$|C_{t}'| \leq k$. Indeed~$|C_{t}'| = |C_{t}| - |(C_{t} \cup F_{t}) \cap S| \leq |C_{t}| \leq k$.
  
  To complete the proof of \cref{claim:antler-sequence} we now derive \cref{item:cutsize}:
  \begin{align*}
   |C_{\leq t}'|
   &= |C_{t}'| + |C_{<t}'|
    &\hspace{-3em}\text{since $C_{t}' \cap C_{<t}' = \emptyset$}\\
   &= |C_{t}| - |(C_{t} \cup F_{t}) \cap S| + |C_{<t}'| \\
   &= |C_{t}| - |(C_{t} \cup F_{t}) \cap S| + \sum_{1 \leq i \leq t-1} (|C_i| - |(C_i \cup F_i) \cap S|)
    &\text{by induction}\\
   &= \sum_{1 \leq i \leq t} (|C_i| - |(C_i \cup F_i) \cap S|).
   \tag*{\qedhere}
  \end{align*}

 \end{claimproof}
 
 To complete the proof of \cref{thm:item:antlers} from \cref{thm:main} we show~$|S^*| = |S \cup S'| \geq |C_{\leq t}|$. By \cref{claim:antler-sequence} we know a $z$-antler-sequence~$C_1',F_1',\ldots,C_t',F_t'$ of width at most~$k$ for~$G'$ exists. Since~$S'$ is obtained from a recursive call we have~$|S'| \geq |C_{\leq t}'|$, so then
  \begin{align*}
   |S \cup S'|
   &= |S| + |S'| \\
   &\geq |S| + |C_{\leq t}'| \\
   &= |S| + \sum_{1 \leq i \leq t} ( |C_i| - |(C_i \cup F_i) \cap S| )
    &\text{by \cref{claim:antler-sequence}}\\
   &= |S| + \left (\sum_{1 \leq i \leq t} |C_i| \right)  - \sum_{1 \leq i \leq t} |(C_i \cup F_i) \cap S| \\
   &= |S| + |C_{\leq t}| - |S \cap (C_{\leq t} \cup F_{\leq t})|
    &\hspace{-3em}\text{since~$C_1, F_1, \ldots, C_t, F_t$ are disjoint}\\
   &\geq |S| + |C_{\leq t}| - |S|\\
   &= |C_{\leq t}|.
   \tag*{\qedhere}
  \end{align*}
\end{proof}

As a corollary to this theorem, we obtain a new type of parameterized-tractability result for \scFVS. For an integer~$z$, let the \emph{$z$-antler complexity} of~$G$ be the minimum number~$k$ for which there exists a (potentially long) sequence~$C_1, F_1, \ldots, C_t, F_t$ of disjoint vertex sets such that for all~$1 \leq i \leq t$, the pair~$(C_{i}, F_{i})$ is a $z$-antler of width at most~$k$ in~$G - (C_{<i} \cup F_{<i})$, and such that~$G - (C_{\leq t} \cup F_{\leq t})$ is acyclic (implying that~$C_{\leq t}$ is a feedback vertex set in~$G$). If no such sequence exists, the $z$-antler complexity of~$G$ is~$+\infty$.

Intuitively, \cref{cor:main} states that optimal solutions can be found efficiently when they are composed out of small pieces, each of which has a low-complexity certificate for belonging to some optimal solution.

\begin{corollary}\label{cor:main}
There is an algorithm that, given a multigraph~$G$, returns an optimal feedback vertex set in time~$f(k^*) \cdot n^{\Oh(z^*)}$, where~$(k^*,z^*)$ is any pair of integers such that the $z^*$-antler complexity of~$G$ is at most~$k^*$.
\end{corollary}
\begin{proof}
Let~$(k^*,z^*)$ be such that the $z^*$-antler complexity of~$G$ is at most~$k^*$. Let~$p_1 \in \Oh(k^5 z^2), p_2 \in \Oh(z)$ be concrete functions such that the running time of \cref{thm:main} is bounded by~$2^{p_1(k,z)} \cdot n^{p_2(z)}$. Consider the pairs~$\{(k',z') \in \setN^2 \mid 1 \leq z' \leq k' \leq n\}$ in order of increasing value of the running-time guarantee~$2^{p_1(k,z)} \cdot n^{p_2(z)}$. For each such pair~$(k',z')$, start from the multigraph~$G$ and invoke \cref{thm:main} to obtain a vertex set~$S$ which is guaranteed to be contained in an optimal solution. If~$G - S$ is acyclic, then~$S$ itself is an optimal solution and we return~$S$. Otherwise we proceed to the next pair~$(k',z')$.

\subparagraph*{Correctness}
The correctness of \cref{thm:main} and the definition of $z$-antler complexity ensure that for~$(k',z') = (k^*,z^*)$, the set~$S$ is an optimal solution. In particular, if~$C_1, F_1, \ldots, C_t, F_t$ is a sequence of vertex sets witnessing that the $z^*$-antler complexity of~$G$ is at most~$k^*$, then \cref{thm:main} is guaranteed by \cref{thm:item:antlers} to output a set~$S$ of size at least~$\sum _{1 \leq i \leq t} |C_i|$, which is equal to the size of an optimal solution on~$G$ by definition. 

\subparagraph*{Running time}
For a fixed choice of~$(k',z')$ the algorithm from \cref{thm:main} runs in time~$2^{\Oh((k')^5 (z')^2)} \cdot n^{\Oh(z')} \leq 2^{\Oh((k^*)^5 (z^*)^2)} \cdot n^{\Oh(z^*)}$ because we try pairs~$(k',z')$ in order of increasing running time. As we try at most~$n^2$ pairs before finding the solution, the corollary follows.
\end{proof}

To conclude, we reflect on the running time of \cref{cor:main} compared to running times of the form~$2^{\Oh(\fvs(G))} \cdot n^{\Oh(1)}$ obtained by FPT algorithms for the parameterization by solution size. If we exhaustively apply \cref{lem:fvc-kernel} with the FVC~$(C,V(G)\setminus C)$, where~$C$ is obtained from a 2-approximation algorithm~\cite{BafnaBF99}, then this gives an \emph{antler-safe} kernelization: it reduces the graph as long as the graph is larger than~$f_r(|C|)$. This opening step reduces the instance size to~$\Oh(\fvs(G)^3)$ without increasing the antler complexity. As observed before, after applying~$\Oh(n^2)$ reduction rules we obtain a graph in which no rules can be applied. This leads to a running time of~$\Oh(n^5)$ of the kernelization. Running \cref{thm:main} to solve the reduced instance yields a total running time of~$2^{\Oh(k^5 z^2)} \fvs(G)^{\Oh(z)} + \Oh(n^5)$. 
This is asymptotically faster than~$2^{\Oh(\fvs(G))}$ when~$z \leq k \in o(\sqrt[7]{\fvs(G)})$ and~$\fvs(G) \in \omega(\log n)$, which captures the intuitive idea sketched above that our algorithmic approach has an advantage when there is an optimal solution that is large but composed of small pieces for which there are low-complexity certificates. The resulting algorithm also behaves fundamentally differently than the known bounded-depth branching algorithms for \textsc{Feedback Vertex Set}, which are forced to explore a deep branching tree for large solutions.

\section{Conclusion} \label{sec:conclusion}
We have taken the first steps into a new direction for preprocessing which aims to investigate how and when a preprocessing phase can guarantee to identify parts of an optimal solution to an \NPhard problem, thereby reducing the running time of the follow-up algorithm. Aside from the technical results concerning antler structures for \scFVS and their algorithmic properties, we consider the conceptual message of this research direction an important contribution of our theoretical work on understanding the power of preprocessing and the structure of solutions to \NPhard problems.

This line of investigation opens up a host of opportunities for future research. For combinatorial problems such as \scVC, \scOCT, and \scDFVS, which kinds of substructures in inputs allow parts of an optimal solution to be identified by an efficient preprocessing phase? Is it possible to give preprocessing guarantees not in terms of the size of an optimal solution, but in terms of measures of the stability~\cite{AngelidakisABCD19,AwasthiBS12,DanielyLS12} of optimal solutions under small perturbations? A recent work shows that preprocessing guarantees in terms of which vertices are \emph{essential} to made a constant-factor approximation is often possible~\cite{BumpusJK22}. 

Some questions also remain open concerning the concrete technical results in this paper. Can the running time of \cref{thm:z:antler} be improved to~$f(k) \cdot n^{\Oh(1)}$? We conjecture that it cannot, but have not been able to prove this.
A related question applies to \scVC: Is there an algorithm running in time~$f(k) \cdot n^{\Oh(1)}$ that, given a graph~$G$ which has disjoint vertex sets~$(C,H)$ such that~$N_G(C) \subseteq H$ and~$H$ of size~$k$ is an optimal vertex cover in~$G[C \cup H]$, outputs a set of size at least~$k$ that is part of an optimal vertex cover in~$G$? (Note that this is an easier target than computing such a decomposition of width~$k$ if one exists, which can be shown to be \Wonehard.)

\begin{figure}
 \centering
 \includegraphics{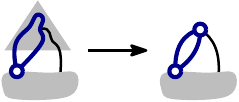}
 \caption{Standard reduction rules for \scFVS reduce any 1-antler of width~$1$ to a pair of vertices with two edges between them, one of which has degree~$3$. Hence we can reduce the graph until all 1-antlers of width~$1$ are removed with the addition of the following reduction rule: If vertices~$u$ and~$v$ are connected by a double edge and~$\degree(v) = 3$ then remove~$u$ and~$v$ from the graph and decrease the solution size by one. These reduction rules can be exhaustively applied in linear time.}
 \label{antler:fig:1-antler-rule}
\end{figure}

To apply the theoretical ideas on antlers in the practical settings that motivated their investigation, it would be interesting to determine which types of antler can be found in \emph{linear} time. A slight extension of the standard reduction rules~\cite[FVS.1--FVS.5]{CyganFKLMPPS15} for \scFVS can be used to detect 1-antlers of width~$1$ in linear time (see \cref{antler:fig:1-antler-rule}). Can the running time of \cref{thm:1:antler} be improved to~$f(k) \cdot (n+m)$? It would also be interesting to investigate which types of antlers are present in practical inputs.

\bibliography{papers}

\end{document}